\documentclass[journal]{IEEEtran}
\usepackage{amsfonts}
\usepackage{amssymb}
\usepackage{amsthm}
\usepackage{amsmath,amsfonts,amssymb}
\usepackage[dvips]{graphicx}
\usepackage{verbatim}
\usepackage{setspace}
\usepackage{bm}
\usepackage{algorithmic} 
\usepackage[ruled,vlined]{algorithm2e}
\usepackage{cite}

\usepackage{changepage}
\usepackage{pdfpages}
\usepackage{color}
\newtheorem{theorem}{Theorem}
\newtheorem{lemma}{Lemma}

\newtheorem{definition}{Definition}

\newcommand{\figwidth}{8}
\IEEEoverridecommandlockouts

\begin{document}
%
\title{Millimeter-Wave Full-Duplex UAV Relay: Joint Positioning, Beamforming, and Power Control}

%
%
%
\author{Lipeng Zhu, ~\IEEEmembership{Student Member,~IEEE,}
        Jun Zhang,
        Zhenyu Xiao,~\IEEEmembership{Senior Member,~IEEE,}
        Xianbin Cao,~\IEEEmembership{Senior Member,~IEEE,}
        Xiang-Gen Xia,~\IEEEmembership{Fellow,~IEEE,}
        and Robert Schober,~\IEEEmembership{Fellow,~IEEE}
\thanks{Manuscript received September 14, 2019; revised January 15, 2020; accepted February 16, 2020. Date of publication XX, 2020; date of current version XX, 2020. This work was supported in part by the
National Natural Science Foundation of China (NSFC) under Grant 61827901, Grant 91538204, and Grant 91738301, and in part by the Open Research Fund of Key Laboratory of Space Utilization, Chinese Academy of Sciences, under Grant LSU-DZXX-2017-02. The associate editor coordinating the review of this paper and approving it for publication was Jiayi Zhang. (\emph{Corresponding authors:Jun Zhang; Zhenyu Xiao}.)}
\thanks{L. Zhu, Z. Xiao and X. Cao are with the School of Electronic and Information Engineering, Beihang University, Beijing 100191, China. (zhulipeng@buaa.edu.cn, xiaozy@buaa.edu.cn, xbcao@buaa.edu.cn)}
\thanks{J. Zhang is with the Advanced Research Institute of Multidisciplinary Science, Beijing Institute of Technology, Beijing 100081, China. (buaazhangjun@vip.sina.com)}
\thanks{X.-G. Xia is with the Department of Electrical and Computer Engineering, University of Delaware, Newark, DE 19716, USA. (xianggen@udel.edu)}
\thanks{R. Schober is with the Institute for Digital Communications, Friedrich-Alexander University of Erlangen-Nuremberg, Erlangen 91054, Germany. (robert.schober@fau.de)}
}


%
%

\maketitle

\begin{abstract}
In this paper, a full-duplex unmanned aerial vehicle (FD-UAV) relay is employed to increase the communication capacity of millimeter-wave (mmWave) networks. Large antenna arrays are equipped at the source node (SN), destination node (DN), and FD-UAV relay to overcome the high path loss of mmWave channels and to help mitigate the self-interference at the FD-UAV relay. Specifically, we formulate a problem for maximization of the achievable rate from the SN to the DN, where the UAV position, analog beamforming, and power control are jointly optimized. Since the problem is highly non-convex and involves high-dimensional, highly coupled variable vectors, we first obtain the conditional optimal position of the FD-UAV relay for maximization of an approximate upper bound on the achievable rate in closed form, under the assumption of a line-of-sight (LoS) environment and ideal beamforming. Then, the UAV is deployed to the position which is closest to the conditional optimal position and yields LoS paths for both air-to-ground links. Subsequently, we propose an alternating interference suppression (AIS) algorithm for the joint design of the beamforming vectors and the power control variables. In each iteration, the beamforming vectors are optimized for maximization of the beamforming gains of the target signals and the successive reduction of the interference, where the optimal power control variables are obtained in closed form. Our simulation results confirm the superiority of the proposed positioning, beamforming, and power control method compared to three benchmark schemes. Furthermore, our results show that the proposed solution closely approaches a performance upper bound for mmWave FD-UAV systems.
\end{abstract}


%
\IEEEpeerreviewmaketitle

\section{Introduction}
\IEEEPARstart{H}{igh} data rates have always been one of the key requirements for wireless mobile communication systems. As the fifth generation (5G) of wireless systems is on the way to deployment, the explosive growth of mobile traffic data poses great challenges in the near future. It is predicted that individual user data rates will exceed 100 Gbps by 2030, and the overall mobile data traffic will reach 5 zettabytes per month \cite{Faisal2019study6G,Saad2019Vision6G,zhangj2019multiantenna,XiaoM2017survmmWave}. In order to meet these tremendous demands, the need for exploiting the high-frequency spectrum is consensus in academia and industry. With its abundant frequency resources, millimeter-wave (mmWave) communication can support gigabit or even terabit transmission rates, which makes it a promising technology for beyond 5G (B5G) and sixth generation (6G) networks \cite{Faisal2019study6G,Saad2019Vision6G,zhangj2019multiantenna,XiaoM2017survmmWave}. Due to the high propagation loss of mmWave signals, beamforming techniques have to be employed to achieve sufficiently high signal-to-noise ratios (SNRs) in mmWave communications \cite{zhangj2018LowADC,niu2015survey,zhangj2019MIMORelay,Gao2016hyb,xiao2018mmWaveNOMA}. Fortunately, benefiting from the small wavelength of mmWave signals, a large number of antennas can be equipped in a small area to realize high array gains \cite{Gao2016hyb,xiao2018mmWaveNOMA,Zhu2019hybNOMA}. Furthermore, the resulting highly directional mmWave beams improve transmission security by reducing the power of the signals received by eavesdroppers \cite{Huang2018HybriPrecod}. However, a drawback of mmWave communications is that obstacles on the ground may prevent the establishment of line-of-sight (LoS) links, which leads to severely attenuated received signal powers even if beamforming is applied. To address this issue, a novel heterogeneous multi-beam cloud radio access network and a decentralized algorithm for beam pair selection were proposed for seamless mmWave coverage in \cite{Liu2018BeamPair}.

On the other hand, unmanned aerial vehicle (UAV) communication has attracted significant attention during the past few years \cite{Zeng2019sky,zengUAVCom,Mozaffari2016UAV,Sun2019Traj,Yu2019frame}, and the integration of UAV into wireless communications is expected to play an important role in B5G and 6G \cite{Zeng2019sky,Li2019UAVSurv}. Benefiting from their mobility, UAVs can be flexibly deployed in areas without infrastructure coverage, e.g., deserts, oceans, and disaster areas where the terrestrial base stations (BSs) may be broken. Compared with conventional terrestrial BSs, UAVs operate at much higher altitudes, and typically have a high probability of being able to establish a line-of-sight (LoS) communication link with the ground user equipment (UE) \cite{Zeng2019sky,zengUAVCom,Zeng2018Traject,Wu2018Traject}. However, UAVs may also suffer from strong interference from neighboring infrastructures/equipments, including neighboring BSs, ground UEs, and other aircrafts. Thus, interference management is one of the key challenges in UAV communications.

To address these problems, the combination of mmWave communications and UAV communications is promising and has unique advantages \cite{xiao2016enabling,zhang2019mmWaveUAV,ZhangC2019mmUAVsur,Gapeyenko2018mmWaveUAV,Wang2019mmWUAVaccess,Zhao2018BeamTrack,Zhu2019UAVBF,Gao2015mmWaveUAV,ZhangW2018track}. First, due to the poor diffraction ability and high propagation loss of mmWave signals, the coverage range of mmWave networks is limited. Energy-efficient UAVs can be flexibly deployed and reconstituted to form a multi-hop network to enlarge the coverage range of mmWave communication networks. Second, at high UAV altitudes, the probability of an LoS link is high because shadowing of the air-to-ground link and the air-to-air link by buildings is unlikely to occur. This property is ideal for the highly directional mmWave signals, for which the non-LoS (NLoS) paths are highly attenuated \cite{xiao2016enabling,ZhangC2019mmUAVsur,Wang2019mmWUAVaccess,Zhu2019UAVBF}. Third, large numbers of antennas can be integrated in the small area available at UAVs because of the small wavelengths of mmWave signals. Hence, directional beamforming can be used to effectively enhance the power of the target signal and to suppress the interference at the UAV.

Motivated by these advantages, integrating UAVs into mmWave cellular has attracted considerable attention recently \cite{xiao2016enabling,zhang2019mmWaveUAV,ZhangC2019mmUAVsur,Gapeyenko2018mmWaveUAV,Wang2019mmWUAVaccess,Zhao2018BeamTrack,Zhu2019UAVBF,Gao2015mmWaveUAV,
ZhangW2018track,zhang2019tracking,zhong2019beampoint,Xu2018ResourceAllo}. In \cite{xiao2016enabling}, the potential of and approaches for combining UAV and mmWave communication were investigated, where fast beamforming training and tracking, spatial division multiple access, blockage, and user discovery were considered. In \cite{ZhangC2019mmUAVsur}, the channel characteristics and precoder design for mmWave-UAV systems were analyzed, and several general challenges and possible solutions were presented for mmWave-UAV cellular networks. The use of UAVs for dynamic routing in mmWave backhaul networks was proposed in \cite{Gapeyenko2018mmWaveUAV}, where the outage probability, spectral efficiency, and outage and non-outage duration distributions were analyzed. In \cite{Wang2019mmWUAVaccess}, multiple access schemes for mmWave-UAV communications were introduced, and a novel link-adaptive constellation-division multiple access technique was proposed. In \cite{Zhao2018BeamTrack}, a blind beam tracking approach was proposed for a UAV-satellite communication system employing a large-scale antenna array. In \cite{zhang2019tracking}, a beam tracking protocol for mmWave UAV-to-UAV communication was designed, where the position and altitude of the UAV were predicted via a Gaussian process based learning algorithm. Due to the unstable beam pointing in mmWave-UAV communications, an optimized beamforming scheme taking into account beam deviation was proposed to overcome beam misalignment in \cite{zhong2019beampoint}. In \cite{Xu2018ResourceAllo}, the two-dimensional position and the downlink beamformer of a fixed-altitude UAV were jointly optimized to mitigate the UAV jittering and user location uncertainty.

Different from the works above, in this paper, we propose to use a full-duplex UAV (FD-UAV) relay to facilitate mmWave communication. Specifically, an FD-UAV relay is deployed between a source node (SN) and a destination node (DN) to establish an LoS link, where large antenna arrays are employed for beamforming to enable directional beams facilitating high channel gains. Although physically separated antenna panels and directional antennas are usually used for mmWave transceivers, the small sidelobes of the radiation pattern, which are inevitable, may result in significant self-interference (SI) for FD relays \cite{Rajagopal2014SImiti,zhangj2018MIMOrelay,xiao2017mmWaveFD,zhang2019EEFD,Yang2018PerfAnalysis,Satyanarayana2019HBFD}. The authors of \cite{Rajagopal2014SImiti} have shown that, in addition to 70-80 dB physical isolation realized by increasing the distance between a transmitter (Tx) antenna panel and an adjacent receiver (Rx) antenna panel, 35-50 dB isolation via SI reduction\footnote{SI reduction methods for FD terminals are usually partitioned into three classes: propagation-domain, analog-circuit-domain, and digital-domain techniques. Tx and Rx beamforming at the FD-UAV relay can be categorized as propagation-domain and analog-circuit-domain approaches, respectively \cite{Sabharwal2014FDWireless,Liu2015FDR}.} is needed to enable successful reception of mmWave signals in in-band FD wireless backhaul links. This motivates us to investigate SI mitigation via mmWave beamforming. In \cite{zhang2019EEFD}, an orthogonal matching pursuit-based (OMP-based) SI-cancellation precoding algorithm was proposed to eliminate the SI and to improve the spectral efficiency in an FD relaying system. In \cite{Yang2018PerfAnalysis}, the impact of the beamwidth and the SI coefficient on the maximum achievable data rate was analyzed for a two-hop amplified-and-forward mmWave relaying system. However, the 3-dimensional (3-D) positioning of the UAV relay, which is investigated in this paper, has not been considered \cite{zhangj2018MIMOrelay,xiao2017mmWaveFD,Satyanarayana2019HBFD,zhang2019EEFD,Yang2018PerfAnalysis}. Besides, the placement, trajectory, resource allocation, and transceiver design of UAVs have also been widely investigated \cite{Zeng2018Traject,Wu2018Traject,Li2019UAVSurv,Alzenad2018UAV,Lyu2017UAV,Liu2019UAVtransceiver,Xu2018ResourceAllo,Sun2019Traj,Yu2019frame}. However, the effects of the mmWave channel and 3-D analog beamforming were not studied in these works. In the considered mmWave communication system, the position of the FD-UAV relay, the beamforming, and the power control have a significant impact on performance. Thus, these variables have to be carefully optimized. The main contributions of this paper can be summarized as follows.

\begin{enumerate}
  \item We propose to deploy an FD-UAV relay to improve the end-to-end performance of a mmWave communication system. We formulate a corresponding optimization problem for maximization of the achievable rate between the SN and the DN. Thereby, Tx and Rx beamforming are utilized to mitigate the SI at the FD-UAV relay. To the best of our knowledge, this is the first work which investigates the joint optimization of positioning, beamforming, and power control for mmWave FD-UAV relays.
  \item To handle the formulated non-convex optimization problem with high-dimensional, highly coupled variable vectors, we first assume an LoS environment and ideal beamforming, where the full array gains can be obtained for the SN-to-UAV (S2V) link and the UAV-to-DN (V2D) link, while the interference can be completely suppressed in the beamforming domain. Based on this assumption, we obtain the corresponding conditional optimal solution for the position of the FD-UAV relay in closed form. Then, we deploy the UAV to the position which is closest to the conditional optimal position and yields LoS paths for both the S2V and the V2D links.
  \item We propose an alternating interference suppression (AIS) algorithm for the joint design of the beamforming vectors (BFVs) and the power control variables. In each iteration, the beam gains for the target signals of the S2V and the V2D links are alternatingly maximized, while the interference is successively reduced. Meanwhile, the optimal power allocation to the SN and FD-UAV relay is updated in closed form for the given position and BFVs.
  \item Simulation results show that the proposed joint positioning, beamforming, and power control scheme outperforms three benchmark schemes. In fact, our results reveal that the proposed joint optimization method can closely approach a performance upper bound for mmWave FD-UAV relay systems.
\end{enumerate}

The rest of this paper is organized as follows. In Section II, we introduce the system model and formulate the proposed joint positioning, beamforming, and power control problem. In Section III, we provide our solution for the formulated problem. Simulation results are presented in Section IV, and the paper is concluded in Section V.

\textit{Notation}: $a$, $\mathbf{a}$, $\mathbf{A}$, and $\mathcal{A}$ denote a scalar, a vector, a matrix, and a set, respectively. $(\cdot)^{\rm{T}}$, $(\cdot)^{*}$, and $(\cdot)^{\rm{H}}$ denote transpose, conjugate, and conjugate transpose, respectively. $|a|$ and $\|\mathbf{a}\|$ denote the absolute value of $a$ and the Frobenius norm of $\mathbf{a}$, respectively. $\lceil a \rceil$ represents the minimum integer no smaller than real number $a$. $\mathbb{E}(\cdot)$ denotes the expected value of a random variable. $\mathfrak{R}(\cdot)$ and $\angle(\cdot)$ denote the real part and the phase of a complex number, respectively. $[\mathbf{a}]_i$ and $[\mathbf{A}]_{i,j}$ denote the $i$-th entry of vector $\mathbf{a}$ and the entry in the $i$-th row and $j$-th column of matrix $\mathbf{A}$, respectively.

\section{System Model and Problem Formulation}
We consider an end-to-end transmission scenario, where a SN serves a remote DN as shown in Fig. \ref{fig:system}\footnote{FD-UAV relays can be used to increase the end-to-end data rate between two ground nodes with poor link quality in B5G mmWave networks. Exemplary application scenarios include BS-to-UE communication, backhaul links \cite{Gapeyenko2018mmWaveUAV}, device-to-device communications \cite{wang2018FD-UAVshare}, and communication between two terrestrial mobile BSs in emergency situations \cite{cao2018airborne}.}. The SN and the DN are equipped with uniform planar arrays (UPAs) employing $N_{\mathrm{S}}^{\mathrm{tot}}=M_{\mathrm{S}} \times N_{\mathrm{S}}$ and $N_{\mathrm{D}}^{\mathrm{tot}}=M_{\mathrm{D}} \times N_{\mathrm{D}}$ antennas, respectively, to overcome the high path loss in the mmWave band. Due to obstacles such as ground buildings, the channel from the SN to the DN may be blocked. Thus, an FD-UAV relay, equipped with an $N_{\mathrm{t}}^{\mathrm{tot}}=M_{\mathrm{t}} \times N_{\mathrm{t}}$ Tx-UPA and an $N_{\mathrm{r}}^{\mathrm{tot}}=M_{\mathrm{r}} \times N_{\mathrm{r}}$ Rx-UPA, is deployed between the SN and the DN to improve system performance.

\begin{figure}[t]
\begin{center}
  \includegraphics[width=\figwidth cm]{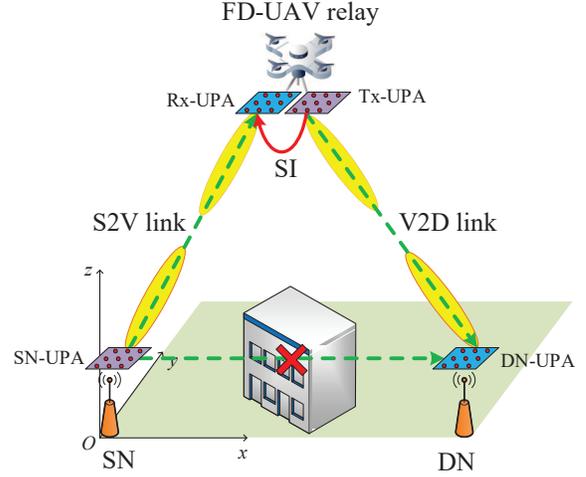}
  \caption{Illustration of the considered mmWave communication system employing an FD-UAV relay to overcome the blockage of the direct mmWave link between the SN and the DN by buildings.}
  \label{fig:system}
\end{center}
\end{figure}

\subsection{Signal Model}
In the considered system, the SN transmits signal $s_{1}$ to the UAV with power $P_{\mathrm{S}}$, and concurrently, the UAV transmits signal $s_{2}$ to the DN with power $P_{\mathrm{V}}$, where $\mathbb{E}(\left | s_{i} \right |^{2})=1$ for $i=1,2$. Thus, the received signal at the UAV is given by\footnote{We assume that a hovering rotary-wing UAV is deployed at a fixed position to support the communication between SN and DN. Thus, the Doppler effect is not considered in this paper.}
\begin{equation}\label{eq_signal_UAV}
\begin{aligned}
&\bar{y}_{\mathrm{V}}=\mathbf{w}_{\mathrm{r}}^{\rm{H}}\mathbf{H}_{\mathrm{S2V}}\mathbf{w}_{\mathrm{S}} \sqrt{P_{\mathrm{S}}}s_{1} + \mathbf{w}_{\mathrm{r}}^{\rm{H}}\mathbf{H}_{\mathrm{SI}}\mathbf{w}_{\mathrm{t}} \sqrt{P_{\mathrm{V}}}s_{2} + n_{1},
\end{aligned}
\end{equation}
where $\mathbf{H}_{\mathrm{S2V}} \in \mathbb{C}^{N_{\mathrm{r}}^{\mathrm{tot}}\times N_{\mathrm{S}}^{\mathrm{tot}}}$ is the channel matrix between the SN and the UAV. $\mathbf{H}_{\mathrm{SI}}\in \mathbb{C}^{N_{\mathrm{r}}^{\mathrm{tot}}\times N_{\mathrm{t}}^{\mathrm{tot}}}$ is the SI channel matrix between the Tx-UPA and the Rx-UPA at the FD-UAV relay. $n_{1}$ denotes the white Gaussian noise at the UAV having zero mean and power $\sigma_1^2$. $\mathbf{w}_{\mathrm{S}}\in \mathbb{C}^{N_{\mathrm{S}}^{\mathrm{tot}}\times 1}$, $\mathbf{w}_{\mathrm{r}}\in \mathbb{C}^{N_{\mathrm{r}}^{\mathrm{tot}}\times 1}$, and $\mathbf{w}_{\mathrm{t}}\in \mathbb{C}^{N_{\mathrm{t}}^{\mathrm{tot}}\times 1}$ represent the SN-BFV, the Rx-BFV at the UAV, and the Tx-BFV at the UAV, respectively.

The received signal at the DN is given by
\begin{equation} \label{eq_signal_DN}
\bar{y}_{\mathrm{D}}=\mathbf{w}_{\mathrm{D}}^{\rm{H}}\mathbf{H}_{\mathrm{S2D}}\mathbf{w}_{\mathrm{S}} \sqrt{P_{\mathrm{S}}}s_{1} + \mathbf{w}_{\mathrm{D}}^{\rm{H}}\mathbf{H}_{\mathrm{V2D}}\mathbf{w}_{\mathrm{t}} \sqrt{P_{\mathrm{V}}}s_{2} + n_{2},
\end{equation}
where $\mathbf{H}_{\mathrm{V2D}}\in \mathbb{C}^{N_{\mathrm{D}}^{\mathrm{tot}}\times N_{\mathrm{t}}^{\mathrm{tot}}}$ is the channel matrix between the UAV and the DN. $\mathbf{H}_{\mathrm{S2D}}\in \mathbb{C}^{N_{\mathrm{D}}^{\mathrm{tot}}\times N_{\mathrm{S}}^{\mathrm{tot}}}$ is the channel matrix between the SN and the DN. $\mathbf{w}_{\mathrm{D}}\in \mathbb{C}^{N_{\mathrm{D}}^{\mathrm{tot}}\times 1}$ denotes the DN-BFV. $n_{2}$ denotes the white Gaussian noise at the DN having zero mean and power $\sigma_2^2$.

In general, there are two main strategies for mmWave beamforming, i.e., digital beamforming and analog beamforming \cite{Gao2016hyb,xiao2018mmWaveNOMA,Zhu2019hybNOMA}. For digital beamforming, each antenna is connected to an independent radio frequency (RF) chain, and thus flexible beamforming is possible due to the large degrees of freedom (DoFs) of the digital beamforming matrices. However, for mmWave systems, the hardware cost and power consumption for digital beamforming are high. In contrast, analog beamforming is more energy efficient, as multiple antennas are connected to only one RF chain via phase shifters. In addition, for FD communication, analog-circuit-domain SI cancellation is usually performed before digital sampling to avoid saturation due to strong SI \cite{Sabharwal2014FDWireless,Liu2015FDR}. For these reasons, analog beamforming is adopted for the considered mmWave FD-UAV relay, which has limited battery capacity and may experience strong SI. The employed analog BFVs impose a constant-modulus (CM) constraint \cite{Gao2016hyb,xiao2018mmWaveNOMA,Zhu2019hybNOMA}, i.e.,

\begin{equation}\label{eq_CM_B}
\left|\left[ \mathbf{w}_{\tau}\right]_{n}\right|=\frac{1}{\sqrt{N_{\tau}^{\mathrm{tot}}}}, ~ 1 \leq n \leq N_{\tau}^{\mathrm{tot}}, \tau=\left\{\mathrm{S}, \mathrm{r}, \mathrm{t}, \mathrm{D}\right\}.
\end{equation}

Then, we can obtain the achievable rates of the S2V and V2D links as follows
\begin{equation}\label{eq_Rate_S2V}
R_{\mathrm{S2V}}=\log_{2}\left (1+ \frac{\left | \mathbf{w}_{\mathrm{r}}^{\rm{H}}\mathbf{H}_{\mathrm{S2V}}\mathbf{w}_{\mathrm{S}} \right |^{2}P_{\mathrm{S}}}{\left | \mathbf{w}_{\mathrm{r}}^{\rm{H}}\mathbf{H}_{\mathrm{SI}}\mathbf{w}_{\mathrm{t}} \right |^{2}P_{\mathrm{V}}+\sigma_{1}^{2}}\right ),
\end{equation}
\begin{equation}\label{eq_Rate_V2D}
R_{\mathrm{V2D}}=\log_{2}\left (1+ \frac{\left | \mathbf{w}_{\mathrm{D}}^{\rm{H}}\mathbf{H}_{\mathrm{V2D}}\mathbf{w}_{\mathrm{t}} \right |^{2}P_{\mathrm{V}}}
{\left | \mathbf{w}_{\mathrm{D}}^{\rm{H}}\mathbf{H}_{\mathrm{S2D}}\mathbf{w}_{\mathrm{S}} \right |^{2}P_{\mathrm{S}}+\sigma_{2}^{2}}\right ).
\end{equation}
Since the S2D link has a small channel gain due to the assumed blockage, the signal received via the S2D link is treated as interference at DN. Note that the achievable rates in \eqref{eq_Rate_S2V} and \eqref{eq_Rate_V2D} hold for coherent detection. Therefore, the FD-UAV relay and DN need to know the effective channel gains $\mathbf{w}_{\mathrm{r}}^{\rm{H}}\mathbf{H}_{\mathrm{S2V}}\mathbf{w}_{\mathrm{S}}$ and $\mathbf{w}_{\mathrm{D}}^{\rm{H}}\mathbf{H}_{\mathrm{V2D}}\mathbf{w}_{\mathrm{t}}$, respectively. The achievable rate between the SN and the DN is the minimum of the rates of the S2V and V2D links, i.e.,
\begin{equation}\label{eq_Rate_S2D}
R_{\mathrm{S2D}}=\min\{R_{\mathrm{S2V}}, R_{\mathrm{V2D}}\}.
\end{equation}

\subsection{Channel Model}
Due to the directivity and sparsity of the far-field mmWave-channel, the channel matrices of the S2V and V2D links can be expressed as a superposition of multipath components, where different paths have different angles of departure (AoDs) and angles of arrival (AoAs). Hence, the channel matrices of the S2V, V2D, and SN-to-DN (S2D) links are modeled as follows \cite{xiao2016enabling,Zhao2018BeamTrack,Zhu2019UAVBF,Gao2016hyb,xiao2018mmWaveNOMA,Zhu2019hybNOMA}
\begin{equation} \label{Channel_S2V}
\begin{aligned}
\mathbf{H}_{\mathrm{S2V}} &= \chi_{\mathrm{S2V}} \beta_{\mathrm{S2V}}^{(0)}
\mathbf{a}_{\mathrm{r}}(\theta_{\mathrm{r}}^{(0)},\phi_{\mathrm{r}}^{(0)})\mathbf{a}_{\mathrm{S}}^{\mathrm{H}}(\theta_{\mathrm{S}}^{(0)},\phi_{\mathrm{S}}^{(0)}) \\ &~~+ \sum \limits_{\ell=1}^{L_{\mathrm{S2V}}} \beta_{\mathrm{S2V}}^{(\ell)}
\mathbf{a}_{\mathrm{r}}(\theta_{\mathrm{r}}^{(\ell)},\phi_{\mathrm{r}}^{(\ell)})\mathbf{a}_{\mathrm{S}}^{\mathrm{H}}(\theta_{\mathrm{S}}^{(\ell)},\phi_{\mathrm{S}}^{(\ell)}),
\end{aligned}
\end{equation}
\begin{equation} \label{Channel_V2D}
\begin{aligned}
\mathbf{H}_{\mathrm{V2D}} &= \chi_{\mathrm{V2D}} \beta_{\mathrm{V2D}}^{(0)}
\mathbf{a}_{\mathrm{D}}(\theta_{\mathrm{D}}^{(0)},\phi_{\mathrm{D}}^{(0)})\mathbf{a}_{\mathrm{t}}^{\mathrm{H}}(\theta_{\mathrm{t}}^{(0)},\phi_{\mathrm{t}}^{(0)}) \\ &~~+ \sum \limits_{\ell=1}^{L_{\mathrm{V2D}}} \beta_{\mathrm{V2D}}^{(\ell)}
\mathbf{a}_{\mathrm{D}}(\theta_{\mathrm{D}}^{(\ell)},\phi_{\mathrm{D}}^{(\ell)})\mathbf{a}_{\mathrm{t}}^{\mathrm{H}}(\theta_{\mathrm{t}}^{(\ell)},\phi_{\mathrm{t}}^{(\ell)}),
\end{aligned}
\end{equation}
\begin{equation} \label{Channel_S2D}
\begin{aligned}
\mathbf{H}_{\mathrm{S2D}}= \sum \limits_{\ell=1}^{L_{\mathrm{S2D}}} \beta_{\mathrm{S2D}}^{(\ell)}
\mathbf{a}_{\mathrm{D}}(\theta_{\mathrm{\widetilde{D}}}^{(\ell)},\phi_{\mathrm{\widetilde{D}}}^{(\ell)})
\mathbf{a}_{\mathrm{S}}^{\mathrm{H}}(\theta_{\mathrm{\widetilde{S}}}^{(\ell)},\phi_{\mathrm{\widetilde{S}}}^{(\ell)}),
\end{aligned}
\end{equation}
where index $\ell=0$ represents the LoS component and indices $\ell \geq 1$ represent the NLoS components. $L_{\mathrm{S2V}}$, $L_{\mathrm{V2D}}$, and $L_{\mathrm{S2D}}$ are the total number of NLoS components for the S2V, V2D, and S2D channels, respectively. Random variables $\chi_{\mathrm{S2V}}$ and $\chi_{\mathrm{V2D}}$ are equal to 1 if the LoS path exists and equal to 0 otherwise. Furthermore, the LoS path from the SN to the DN is assumed to be blocked, which is the main motivation for deploying an FD-UAV relay. $\beta_{\mathrm{S2V}}^{(\ell)}$, $\beta_{\mathrm{V2D}}^{(\ell)}$, and $\beta_{\mathrm{S2D}}^{(\ell)}$ are the complex coefficients of the S2V, V2D, and S2D paths, respectively. $\theta_{\mathrm{S}}^{(\ell)}$, $\phi_{\mathrm{S}}^{(\ell)}$, $\theta_{\mathrm{r}}^{(\ell)}$, and $\phi_{\mathrm{r}}^{(\ell)}$ represent the elevation AoD (E-AoD), azimuth AoD (A-AoD), elevation AoA (E-AoA), and azimuth AoA (A-AoA) of the S2V path, respectively. $\theta_{\mathrm{t}}^{(\ell)}$, $\phi_{\mathrm{t}}^{(\ell)}$, $\theta_{\mathrm{D}}^{(\ell)}$, and $\phi_{\mathrm{D}}^{(\ell)}$ represent the E-AoD, A-AoD, E-AoA, and A-AoA of the V2D path, respectively. $\theta_{\mathrm{\widetilde{B}}}^{(\ell)}$, $\phi_{\mathrm{\widetilde{B}}}^{(\ell)}$, $\theta_{\mathrm{\widetilde{U}}}^{(\ell)}$, and $\phi_{\mathrm{\widetilde{U}}}^{(\ell)}$ represent the E-AoD, A-AoD, E-AoA, and A-AoA of the S2D path, respectively. $\mathbf{a}_{\mathrm{S}}(\cdot)$, $\mathbf{a}_{\mathrm{r}}(\cdot)$, $\mathbf{a}_{\mathrm{t}}(\cdot)$, and $\mathbf{a}_{\mathrm{D}}(\cdot)$ are the steering vectors of the UPA at the SN, the Rx-UPA at the FD-UAV relay, the Tx-UPA at the FD-UAV relay, and the UPA at the DN, respectively. The steering vectors are given as follows \cite{balanis2016antenna}
\begin{equation} \label{eq_steeringVCT}
\begin{aligned}
&\mathbf{a}_{\tau}(\theta_{\tau},\phi_{\tau})=
[1,\cdots,e^{j2\pi\frac{d}{\lambda}\cos\theta_{\tau}[(m-1)\cos \phi_{\tau}+(n-1)\sin \phi_{\tau}]},\\
&~~~~\cdots,e^{j2\pi\frac{d}{\lambda}\cos\theta_{\tau}[(M_{\tau}^{\mathrm{tot}}-1)\cos \phi_{\tau}+(N_{\tau}^{\mathrm{tot}}-1)\sin \phi_{\tau}]} ]^{\mathrm{T}},
\end{aligned}
\end{equation}
where $d$ is the spacing between adjacent antennas, $\lambda$ is the carrier wavelength, $0 \leq m \leq M_{\tau}^{\mathrm{tot}}-1$, $0 \leq n \leq N_{\tau}^{\mathrm{tot}}-1$, and $\tau=\left\{\mathrm{S}, \mathrm{r}, \mathrm{t}, \mathrm{D}\right\}$. Particularly, for half-wavelength spacing arrays, we have $d=\lambda/2$.

For the LoS path of the SI channel at the FD-UAV relay, the far-field range condition, $R \geq 2D^2/\lambda$, where $R$ is the distance between the Tx antenna and the Rx antenna and $D$ is the diameter of the antenna aperture, does not hold in general. Thus, the SI channel has to be modeled using the near-field model as follows \cite{xiao2017mmWaveFD,Satyanarayana2019HBFD,zhang2019EEFD}
\begin{equation}\label{Channel_SI}
\left[\mathbf{H}_{\mathrm{SI}}\right]_{m, n}=\beta_{\mathrm{SI}}^{(m,n)} \exp \left(-j 2 \pi \frac{r_{m, n}}{\lambda}\right),
\end{equation}
where $\beta_{\mathrm{SI}}^{(m,n)}$ are the complex coefficients of the SI channel, and $r_{m,n}$ is the distance between the $m$-th Tx array element and the $n$-th Rx array element. Note that for the SI channel, NLoS paths may also exist, due to reflectors around the FD-UAV relay. Since the propagation distances of the NLoS paths are much longer than that of the LoS path, which leads to a higher attenuation, we focus on the LoS component of the SI channel \cite{xiao2017mmWaveFD,Satyanarayana2019HBFD,zhang2019EEFD}. Although the SI channel model is more complicated compared to the far-field channel model, the FD-UAV relay is expected to be able to acquire the corresponding channel state information (CSI), as the SI channel is only slowly varying \cite{xiao2017mmWaveFD}. In this paper, we assume that for a given fixed position of the FD-UAV relay, instantaneous CSI is available at the SN, FD-UAV relay, and DN via channel estimation. However, the FD-UAV can acquire only the CSI for the position it is at.

Next, we provide the models for the parameters of the channel matrices in \eqref{Channel_S2V}-\eqref{Channel_S2D}, \eqref{Channel_SI}. As shown in Fig. \ref{fig:system}, we establish a coordinate system with the origin at the SN, and the three axes $x$, $y$, and $z$, are separately aligned with the directions of east, north, and vertical (upward), respectively. Without loss of generality, we assume the SN and the DN both have zero altitude, and the UPAs are parallel to the plane spanned by the $x$ and $y$ axes. Then, the coordinates of the DN are $(x_{\mathrm{D}},y_{\mathrm{D}},0)$, and the coordinates of the FD-UAV relay are $(x_{\mathrm{V}},y_{\mathrm{V}},h_{\mathrm{V}})$.

According to basic geometry, we obtain the parameters of the S2V link, including the distance and the AoDs and AoAs of the LoS path, as follows
\begin{equation} \label{anlge_S2V}
\left \{
\begin{aligned}
&{d_{\mathrm{S} 2 \mathrm{V}}=\sqrt{x_{\mathrm{V}}^{2}+y_{\mathrm{V}}^{2}+h_{\mathrm{V}}^{2}}}, \\
&{\theta_{\mathrm{S}}^{(0)}=\theta_{\mathrm{r}}^{(0)}=\arctan \frac{h_{\mathrm{V}}}{\sqrt{x_{\mathrm{V}}^{2}+y_{\mathrm{V}}^{2}}}}, \\
&{\phi_{\mathrm{S}}^{(0)}=\phi_{\mathrm{r}}^{(0)}=\arctan \frac{y_{\mathrm{V}}}{x_{\mathrm{V}}}}.
\end{aligned}
\right.
\end{equation}
Similarly, we obtain the parameters of the V2D link as
\begin{equation} \label{anlge_V2D}
\left \{
\begin{aligned}
&{d_{\mathrm{V} 2 \mathrm{D}}=\sqrt{\left(x_{\mathrm{V}}-x_{\mathrm{D}}\right)^{2}+\left(y_{\mathrm{V}}-y_{\mathrm{D}}\right)^{2}+h_{\mathrm{V}}^{2}}}, \\
&{\theta_{\mathrm{t}}^{(0)}=\theta_{\mathrm{D}}^{(0)}=\arctan \frac{h_{\mathrm{V}}}{\sqrt{\left(x_{\mathrm{V}}-x_{\mathrm{D}}\right)^{2}+\left(y_{\mathrm{V}}-y_{\mathrm{D}}\right)^{2}}}}, \\
&{\phi_{\mathrm{t}}^{(0)}=\phi_{\mathrm{D}}^{(0)}=\arctan \frac{y_{\mathrm{V}}-y_{\mathrm{D}}}{x_{\mathrm{V}}-x_{\mathrm{D}}}}.
\end{aligned}
\right.
\end{equation}
For the S2V, V2D, and S2D links, which are characterized by far-field channels, the AoDs and AoAs of the NLoS paths are assumed to be uniformly distributed. Considering the propagation conditions at mmWave frequencies, the complex coefficients of the LoS and NLoS paths are modeled as \cite{Rappaport2015meas}
\begin{equation} \label{coef_LoS}
\beta_{\mathrm{S2V}}^{(0)}=\frac{c}{4\pi f_{c}}d_{\mathrm{S2V}}^{-\alpha_{\mathrm{LoS}}/2}, \beta_{\mathrm{V2D}}^{(0)}=\frac{c}{4\pi f_{c}}d_{\mathrm{V2D}}^{-\alpha_{\mathrm{LoS}}/2},
\end{equation}
\begin{equation} \label{coef_NLoS}
\left \{
\begin{aligned}
\beta_{\mathrm{S2V}}^{(\ell)}=\frac{c}{4\pi f_{c}}d_{\mathrm{S2V}}^{-\alpha_{\mathrm{NLoS}}/2}X_{1}, ~~\text{for}~ \ell\geq 1,\\ 
\beta_{\mathrm{V2D}}^{(\ell)}=\frac{c}{4\pi f_{c}}d_{\mathrm{V2D}}^{-\alpha_{\mathrm{NLoS}}/2}X_{2}, ~~\text{for}~ \ell\geq 1,\\ 
\beta_{\mathrm{S2D}}^{(\ell)}=\frac{c}{4\pi f_{c}}d_{\mathrm{S2D}}^{-\alpha_{\mathrm{NLoS}}/2}X_{3}, ~~\text{for}~ \ell\geq 1,
\end{aligned}
\right.
\end{equation}
where $c$ is the constant speed of light, $f_{c}$ is the carrier frequency, and $d_{\mathrm{S2D}}=\sqrt{x_{\mathrm{D}}^{2}+y_{\mathrm{D}}^{2}}$ is the distance of the S2D link. $\alpha_{\mathrm{LoS}}$ and $\alpha_{\mathrm{NLoS}}$ are the large-scale path loss exponents for the LoS and NLoS links, respectively. $X_{i}$, $i=1,2,3$, are the gains for the NLoS paths, which are assumed to be circular symmetric complex Gaussian random variables with zero mean and standard deviation $\sigma_{f}$, i.e., Rayleigh fading is assumed \cite{TseFundaWC}. For the SI channel, the complex coefficient is given by \cite{xiao2017mmWaveFD,Satyanarayana2019HBFD,zhang2019EEFD}
\begin{equation} \label{coef_SI}
\beta_{\mathrm{SI}}^{(m,n)}=\frac{c}{4\pi f_{c}}r_{m,n}^{-\alpha_{\mathrm{LoS}}/2}.
\end{equation}

Besides, due to obstacles on the ground, the probabilities that an LoS path exists for the S2V and V2D links are modelled as logistic functions of the elevation angles \cite{Hourani2014MaxCov}, i.e.,
\begin{equation} \label{prob_LoS_S2V}
\hat{P}_{\mathrm{S2V}}^{\mathrm{LoS}}=\frac{1}{1+a\exp{(-b(\frac{180}{\pi}\theta_{\mathrm{r}}^{(0)}-a)})},
\end{equation}
\begin{equation} \label{prob_LoS_V2D}
\hat{P}_{\mathrm{V2D}}^{\mathrm{LoS}}=\frac{1}{1+a\exp{(-b(\frac{180}{\pi}\theta_{\mathrm{t}}^{(0)}-a)})},
\end{equation}
where $a$ and $b$ are positive modelling parameters whose values depend on the propagation environment. Random variables $\chi_{\mathrm{S2V}}$ and $\chi_{\mathrm{V2D}}$ in \eqref{Channel_S2V} and \eqref{Channel_V2D} are generated based on the LoS probabilities in \eqref{prob_LoS_S2V} and \eqref{prob_LoS_V2D}, respectively. Hereto, the statistical channel models for S2V, V2D, and S2D links have been provided. For the communication scenario considered in this paper, the instantaneous channel responses are generated according to these statistical models.

From the above, we observe that the S2V and V2D channels, including the propagation loss, the spatial angles, and the probabilities that an LoS link exists, depend on the position of the UAV. Thus, the position of the FD-UAV relay has significant influence on the achievable data rate. However, in practice, the instantaneous CSI is not a priori known by the SN, UAV, and DN before the UAV is deployed at a given fixed position and performs channel estimation. This property distinguishes the considered FD-UAV relay system from traditional FD relay networks on the ground where the position of the relay is fixed.

\subsection{Problem Formulation}
To maximize the achievable rate from the SN to the DN, we formulate the following problem for joint optimization of the UAV positioning, BFVs, and transmit powers:
\begin{equation}\label{eq_problem}
\begin{aligned}
\mathop{\mathrm{Maximize}}\limits_{\Psi}~~ &\min \left\{R_{\mathrm{S2V}}, R_{\mathrm{V2D}}\right\}\\
\mathrm{Subject~ to}~~ &\left(x_{\mathrm{V}}, y_{\mathrm{V}}\right) \in\left[0, x_{\mathrm{D}}\right] \times\left[ 0, y_{\mathrm{D}} \right], \\
&h_{\min} \leq h_{\mathrm{V}} \leq h_{\max},\\
&0 \leq P_{\mathrm{S}} \leq P_{\mathrm{S}}^{\mathrm{tot}}, \\
&0 \leq P_{\mathrm{V}} \leq P_{\mathrm{V}}^{\mathrm{tot}}, \\
&\left|\left[ \mathbf{w}_{\tau}\right]_{n}\right|=\frac{1}{\sqrt{N_{\tau}^{\mathrm{tot}}}}, ~ \tau=\left\{\mathrm{S}, \mathrm{r}, \mathrm{t}, \mathrm{D}\right\}, ~\forall n,
\end{aligned}
\end{equation}
where $\Psi=\{x_{\mathrm{V}}, y_{\mathrm{V}}, h_{\mathrm{V}}, \mathbf{w}_{\mathrm{S}}, \mathbf{w}_{\mathrm{D}}, \mathbf{w}_{\mathrm{r}}, \mathbf{w}_{\mathrm{t}}, P_{\mathrm{S}}, P_{\mathrm{V}}\}$. The first constraint indicates that the FD-UAV relay should be deployed between the SN and the DN. The second constraint limits the altitude of the FD-UAV relay, where $h_{\min}$ and $h_{\max}$ are the minimum and maximum values, respectively. The third and fourth constraints indicate that the transmit powers are nonnegative and cannot exceed a maximum value, where $P_{\mathrm{S}}^{\mathrm{tot}}$ and $P_{\mathrm{V}}^{\mathrm{tot}}$ are the maximum transmit powers of the SN and the FD-UAV relay, respectively. The fifth constraint is the CM constraint on the analog BFVs. Due to the non-convex nature and high-dimensional, highly coupled variable vectors, Problem \eqref{eq_problem} cannot be directly solved with existing optimization tools. Thus, we develop a solution for \eqref{eq_problem} in the next section.

\section{Solution of the Problem}
Since in Problem \eqref{eq_problem} the position variables, BFVs, and power control variables are highly coupled, it is difficult to obtain a globally optimal solution. In this section, we develop a sub-optimal solution for Problem \eqref{eq_problem}. Since the position of the FD-UAV relay crucially affects the S2V and V2D channel matrices, we first optimize $x_{\mathrm{V}}$, $y_{\mathrm{V}}$, and $h_{\mathrm{V}}$. Then, given the position of the FD-UAV relay and the corresponding instantaneous CSI, we develop the proposed AIS algorithm for joint optimization of the BFVs and the power control variables. Finally, we summarize the proposed overall solution for joint positioning, beamforming, and power control in mmWave FD-UAV relay systems.
\subsection{Positioning Under Ideal Beamforming}
Since the LoS path is much stronger than the NLoS paths at mmWave frequencies in general, we neglect the NLoS paths for optimization of the position of the FD-UAV relay in this subsection. Furthermore, the motivation for deploying an FD-UAV relay is to establish LoS communication links for both the S2V and the V2D links, otherwise the communication quality will be poor. Thus, we assume that both the S2V and the V2D links have an LoS path\footnote{For a sufficiently large $h_{\min}$, the probabilities that LoS paths exist, given by \eqref{prob_LoS_S2V} and \eqref{prob_LoS_V2D}, approach 1 \cite{Zeng2019sky}, and thus the LoS-environment assumption adopted for positioning is reasonable. If an LoS path does not exist for the S2V and/or the V2D links at the optimized position, we resort to the strategy specified after Theorem 1.}, and optimize the position of the FD-UAV relay under the assumption of ideal beamforming.

\begin{definition} \label{Defi_idealBF} \textbf{(Ideal Beamforming)}
For ideal BFVs $\mathbf{w}_{\tau}$, $\tau=\left\{\mathrm{S}, \mathrm{r}, \mathrm{t}, \mathrm{D}\right\}$, assuming an LoS environment, the FD-UAV relay system achieves the full array gains for the S2V and V2D links, respectively, while the SI and the interference caused by the S2D link are completely eliminated in the beamforming domain, i.e.,
\begin{equation}\label{eq_idealBF}
\left\{
\begin{aligned}
&\left | \mathbf{w}_{\mathrm{r}}^{\mathrm{H}}\mathbf{H}_{\mathrm{S2V}}\mathbf{w}_{\mathrm{S}} \right |^{2}=\left|\beta_{\mathrm{S2V}}^{(0)}\right|^{2} N_{\mathrm{S}}^{\mathrm{tot}} N_{\mathrm{r}}^{\mathrm{tot}},\\
&\left | \mathbf{w}_{\mathrm{D}}^{\mathrm{H}}\mathbf{H}_{\mathrm{V2D}}\mathbf{w}_{\mathrm{t}} \right |^{2}=\left|\beta_{\mathrm{V2D}}^{(0)}\right|^{2} N_{\mathrm{t}}^{\mathrm{tot}} N_{\mathrm{D}}^{\mathrm{tot}},\\
&\left | \mathbf{w}_{\mathrm{r}}^{\mathrm{H}}\mathbf{H}_{\mathrm{SI}}\mathbf{w}_{\mathrm{t}} \right |^{2}=\left | \mathbf{w}_{\mathrm{D}}^{\rm{H}}\mathbf{H}_{\mathrm{S2D}}\mathbf{w}_{\mathrm{S}} \right |^{2}=0.
\end{aligned}
\right.
\end{equation}
\end{definition}
Substituting \eqref{coef_LoS} and \eqref{eq_idealBF} into \eqref{eq_Rate_S2V} and \eqref{eq_Rate_V2D}, for a pure LoS environment, we obtain upper bounds for the achievable rates of the S2V and V2D links as follows
\begin{equation}\label{eq_RateBound_S2V}
\bar{R}_{\mathrm{S2V}}=\log _{2}\left(1+\frac{c^{2}}{16\pi^{2} f_{c}^{2}} \frac{N_{\mathrm{S}}^{\mathrm{tot}} N_{\mathrm{r}}^{\mathrm{tot}} P_{\mathrm{S}}^{\mathrm{tot}}}{d_{\mathrm{S2V}}^{\alpha_{\mathrm{LoS}}} \sigma_{1}^{2}}\right),
\end{equation}
\begin{equation}\label{eq_RateBound_V2D}
\bar{R}_{\mathrm{V2D}}=\log _{2}\left(1+\frac{c^{2}}{16\pi^{2} f_{c}^{2}} \frac{N_{\mathrm{t}}^{\mathrm{tot}} N_{\mathrm{D}}^{\mathrm{tot}} P_{\mathrm{V}}^{\mathrm{tot}}}{d_{\mathrm{V2D}}^{\alpha_{\mathrm{LoS}}} \sigma_{2}^{2}}\right).
\end{equation}
Note that the upper bounds given by \eqref{eq_RateBound_S2V} and \eqref{eq_RateBound_V2D} are valid for a pure LoS environment without NLoS paths. When the NLoS paths are also considered, we obtain upper bounds for the achievable rates of the S2V and V2D links as follows
\begin{equation}\label{eq_RateBound2_S2V}
\bar{\bar{R}}_{\mathrm{S2V}}=\log _{2}\left(1+\sum \limits_{\ell=0}^{L_{\mathrm{S2V}}} \left|\beta_{\mathrm{S2V}}^{(\ell)}\right|^{2} \frac{N_{\mathrm{S}}^{\mathrm{tot}} N_{\mathrm{r}}^{\mathrm{tot}} P_{\mathrm{S}}^{\mathrm{tot}}}{\sigma_{1}^{2}}\right),
\end{equation}
\begin{equation}\label{eq_RateBound2_V2D}
\bar{\bar{R}}_{\mathrm{V2D}}=\log _{2}\left(1+\sum \limits_{\ell=0}^{L_{\mathrm{V2D}}} \left|\beta_{\mathrm{V2D}}^{(\ell)}\right|^{2} \frac{N_{\mathrm{t}}^{\mathrm{tot}} N_{\mathrm{D}}^{\mathrm{tot}} P_{\mathrm{V}}^{\mathrm{tot}}}{\sigma_{2}^{2}}\right).
\end{equation}
We refer to the achievable rates in \eqref{eq_RateBound_S2V} and \eqref{eq_RateBound_V2D} as \emph{approximate upper bounds}, and to the achievable rates in \eqref{eq_RateBound2_S2V} and \eqref{eq_RateBound2_V2D} as \emph{strict upper bounds}. Since the NLoS paths are not a priori known for different positions of the FD-UAV relay, the approximate upper bounds are used for UAV positioning. The performance gap between the approximate upper bounds and the strict upper bounds will be evaluated via simulations in Section IV.

As can be seen, for an LoS environment and ideal beamforming, the achievable rates in \eqref{eq_RateBound_S2V} and \eqref{eq_RateBound_V2D} depend only on the distances $d_{\mathrm{S2V}}$, $d_{\mathrm{V2D}}$, and the transmit powers $P_{\mathrm{S}}$, $P_{\mathrm{V}}$. Note that the achievable rates are both monotonically increasing in the transmit power. Hence, $P_{\mathrm{S}}^{\mathrm{tot}}$ and $P_{\mathrm{V}}^{\mathrm{tot}}$ are the optimal transmit powers maximizing the upper-bound rate for an LoS environment and ideal beamforming. In the following theorem, we provide the corresponding optimal position of the FD-UAV relay.
\begin{theorem} \label{Theo_posi}
For an LoS environment and ideal beamforming, the optimal solution for the UAV's position is given by $\left(x_{\mathrm{V}}^{\star},y_{\mathrm{V}}^{\star},h_{\mathrm{V}}^{\star}\right) = \left(\rho^{\star}x_{\mathrm{D}},\rho^{\star}y_{\mathrm{D}},h_{\min}\right)$ with
\begin{equation} \label{OptPosition}
\rho^{\star}=\left \{
\begin{aligned}
&0,~\text{if}~\frac{N_{\mathrm{S}}^{\mathrm{tot}} N_{\mathrm{r}}^{\mathrm{tot}} P_{\mathrm{S}}^{\mathrm{tot}}\sigma_{2}^{2}}{N_{\mathrm{t}}^{\mathrm{tot}} N_{\mathrm{D}}^{\mathrm{tot}} P_{\mathrm{V}}^{\mathrm{tot}}\sigma_{1}^{2}} \leq \frac{h_{\min}^{\alpha_{\mathrm{LoS}}}}{\left(x_{\mathrm{D}}^{2}+y_{\mathrm{D}}^{2}+h_{\min}^{2}\right)^{\frac{\alpha_{\mathrm{LoS}}}{2}}}, \\
&1,~\text{if}~\frac{N_{\mathrm{S}}^{\mathrm{tot}} N_{\mathrm{r}}^{\mathrm{tot}} P_{\mathrm{S}}^{\mathrm{tot}}\sigma_{2}^{2}}{N_{\mathrm{t}}^{\mathrm{tot}} N_{\mathrm{D}}^{\mathrm{tot}} P_{\mathrm{V}}^{\mathrm{tot}}\sigma_{1}^{2}} \geq \frac{\left(x_{\mathrm{D}}^{2}+y_{\mathrm{D}}^{2}+h_{\min}^{2}\right)^{\frac{\alpha_{\mathrm{LoS}}}{2}}}{h_{\min}^{\alpha_{\mathrm{LoS}}}}, \\
&\frac{1}{2},~\text{if}~\frac{N_{\mathrm{S}}^{\mathrm{tot}} N_{\mathrm{r}}^{\mathrm{tot}} P_{\mathrm{S}}^{\mathrm{tot}}\sigma_{2}^{2}}{N_{\mathrm{t}}^{\mathrm{tot}} N_{\mathrm{D}}^{\mathrm{tot}} P_{\mathrm{V}}^{\mathrm{tot}}\sigma_{1}^{2}}=1,\\
&\frac{-b'-\sqrt{b'^2-4a'c'}}{2a'},~\text{otherwise},
\end{aligned}
\right.
\end{equation}
where parameters $a'$, $b'$, and $c'$ are given by
\begin{equation}
\left \{
\begin{aligned}
&a'=\left(\left(\frac{N_{\mathrm{S}}^{\mathrm{tot}} N_{\mathrm{r}}^{\mathrm{tot}} P_{\mathrm{S}}^{\mathrm{tot}}}{\sigma_{1}^{2}}\right)^{\frac{2}{\alpha_{\mathrm{LoS}}}}-\left(\frac{N_{\mathrm{t}}^{\mathrm{tot}} N_{\mathrm{D}}^{\mathrm{tot}} P_{\mathrm{V}}^{\mathrm{tot}}}{\sigma_{2}^{2}}\right)^{\frac{2}{\alpha_{\mathrm{LoS}}}}\right)\\
&~~~~~~~~~\times \left(x_{\mathrm{D}}^{2}+y_{\mathrm{D}}^{2}\right),\\
&b'=-2\left(\frac{N_{\mathrm{S}}^{\mathrm{tot}} N_{\mathrm{r}}^{\mathrm{tot}} P_{\mathrm{S}}^{\mathrm{tot}}}{\sigma_{1}^{2}}\right)^{\frac{2}{\alpha_{\mathrm{LoS}}}}\left(x_{\mathrm{D}}^{2}+y_{\mathrm{D}}^{2}\right),\\
&c'=\left(\frac{N_{\mathrm{S}}^{\mathrm{tot}} N_{\mathrm{r}}^{\mathrm{tot}} P_{\mathrm{S}}^{\mathrm{tot}}}{\sigma_{1}^{2}}\right)^{\frac{2}{\alpha_{\mathrm{LoS}}}}\left(x_{\mathrm{D}}^{2}+y_{\mathrm{D}}^{2}\right)+\\
&\left(\left(\frac{N_{\mathrm{S}}^{\mathrm{tot}} N_{\mathrm{r}}^{\mathrm{tot}} P_{\mathrm{S}}^{\mathrm{tot}}}{\sigma_{1}^{2}}\right)^{\frac{2}{\alpha_{\mathrm{LoS}}}}-\left(\frac{N_{\mathrm{t}}^{\mathrm{tot}} N_{\mathrm{D}}^{\mathrm{tot}} P_{\mathrm{V}}^{\mathrm{tot}}}{\sigma_{2}^{2}}\right)^{\frac{2}{\alpha_{\mathrm{LoS}}}}\right) h_{\min}^{2}.
\end{aligned}
\right.
\end{equation}
\end{theorem}
\begin{proof}
See Appendix A.
\end{proof}
Since an LoS environment and ideal beamforming are assumed in Theorem \ref{Theo_posi}, in the following, we refer to \eqref{OptPosition} as the \emph{conditional optimal position} of the FD-UAV relay. However, due to possible obstacles on the ground, the LoS path for the S2V and V2D links may be blocked. Since the existence of an LoS path depends on the actual environment and is not a priori known by the SN, UAV, and DN, it is necessary for the FD-UAV relay to adjust its position if needed. To this end, the UAV is initially deployed to the conditional optimal position $\left(x_{\mathrm{V}}^{\star},y_{\mathrm{V}}^{\star},h_{\mathrm{V}}^{\star}\right)$ and the instantaneous CSI is acquired. If there exist LoS paths for both the S2V and the V2D links, the UAV remains at position $\left(x_{\mathrm{V}}^{\star},y_{\mathrm{V}}^{\star},h_{\mathrm{V}}^{\star}\right)$ as it is optimal for an LoS environment. Otherwise, if an LoS path for the S2V link and/or the V2D link does not exist for position $\left(x_{\mathrm{V}}^{\star},y_{\mathrm{V}}^{\star},h_{\mathrm{V}}^{\star}\right)$, the UAV moves around the initial position until LoS links are established. Specifically, we start an iterative process indexed by $t$. The $t$-th neighborhood for the position of the FD-UAV relay is defined as $\mathcal{C}_{t}=\{(x_{\mathrm{V}}^{\star} \pm i\epsilon_{x}, y_{\mathrm{V}}^{\star} \pm j\epsilon_{y}, h_{\min}+k\epsilon_{h}) \in \mathcal{C} \mid i,j,k=0,1,\cdots,t\}$, where $\epsilon_{x}$, $\epsilon_{y}$, and $\epsilon_{h}$ determine the granularity of the search space for directions $x$, $y$, and $z$, respectively. $\mathcal{C}=[0, x_{\mathrm{D}}]\times[0, y_{\mathrm{D}}]\times[h_{\min}, h_{\max}]$ denotes the feasible region for the position of the FD-UAV relay. During the search, the UAV gradually increases its distance from $\left(x_{\mathrm{V}}^{\star},y_{\mathrm{V}}^{\star},h_{\mathrm{V}}^{\star}\right)$, i.e., index $t$ is increased by 1 in each iteration. The iteration terminates when a point in $\mathcal{C}_{t}$ is found which yields LoS paths for both the S2V and the V2D links, and the selected position of the FD-UAV relay is given by
\begin{equation} \label{AdjPosition}
\left(x_{\mathrm{V}}^{\circ},y_{\mathrm{V}}^{\circ},h_{\mathrm{V}}^{\circ}\right)=\mathop{\mathrm{arg~min}}\limits_{(x,y,h)\in \mathcal{L}_{t}} d_{x,y,h},
\end{equation}
where $\mathcal{L}_{t} \subseteq \mathcal{C}_{t} \setminus \mathcal{C}_{t-1}$ denotes the set of coordinates which yield LoS paths for both the S2V and the V2D links in the $t$-th neighborhood, and $\mathcal{C}_{t} \setminus \mathcal{C}_{t-1}$  contains the elements of $\mathcal{C}_{t}$ that are not included in $\mathcal{C}_{t-1}$. $d_{x,y,h}=\sqrt{(x-x_{\mathrm{V}}^{\star})^{2}+(y-y_{\mathrm{V}}^{\star})^{2}+(h-h_{\mathrm{V}}^{\star})^{2}}$ is the Euclidean distance between the candidate coordinates $(x,y,h)$ and $\left(x_{\mathrm{V}}^{\star},y_{\mathrm{V}}^{\star},h_{\mathrm{V}}^{\star}\right)$. If $\mathcal{L}_{t}$ contains multiple sets of coordinates which have the smallest distance from the initial position, one set of the coordinates is selected at random from these candidates.

Hereto, the position of the FD-UAV relay is determined. Note that the transmit powers at the SN and FD-UAV relay are set to the maximal possible values. However, this may result in a waste of power. For instance, when the achievable rate of the S2V link is always smaller than that of the V2D link, increasing the FD-UVA's transmit power can not enlarge the achievable rate of the DN because the rate is limited by the S2V link. Besides, if the SI is not completely suppressed for non-ideal beamforming, increasing the FD-UAV's transmit power may also increase the interference for the S2V link, and thus the achievable rate decreases. For these reasons, in the following, we first design the BFVs before we optimize the power control to maximize the achievable rate.

\subsection{Beamforming Design}
In this subsection, we design the BFVs for the given coordinates of the FD-UAV relay. It is assumed that full CSI is available at the SN, the DN, and the FD-UAV relay, where both the LoS and NLoS components are considered for the S2V and the V2D links. Due to the non-convex CM constraints and the coupled variables, it is challenging to jointly optimize the BFVs at the SN, UAV, and DN. To address this issue, we propose the AIS algorithm, which employs alternating optimization to design the BFV at the SN, the BFV at the DN, and the Tx/Rx-BFV at the FD-UAV relay. First, we initialize the BFVs with the normalized steering vectors corresponding to the LoS paths for the S2V and V2D channels, i.e.,
\begin{equation} \label{eq_BFV_SN}
\mathbf{w}_{\tau}^{(0)}=\frac{1}{\sqrt{N_{\tau}^{\mathrm{tot}}}}\mathbf{a}_{\tau}(\theta_{\tau}^{(0)},\phi_{\tau}^{(0)}), \tau=\left\{\mathrm{S}, \mathrm{r}, \mathrm{t}, \mathrm{D}\right\}.
\end{equation}

Then, we start an iterative process. Given an SN-BFV, a DN-BFV, and a Tx-BFV, such that the received signal power of the V2D link and the interference from the S2D link are fixed, motivated by \eqref{eq_Rate_V2D}, we optimize the Rx-BFV to maximize the received signal power of the S2V link, while suppressing the SI. Specifically, in the $k$-th iteration, we solve the following problem:
\begin{equation}\label{eq_problem_sub1}
\begin{aligned}
\mathop{\mathrm{Maximize}}\limits_{\mathbf{w}_{\mathrm{r}}}~~~~~ &\left|\mathbf{w}_{\mathrm{r}}^{\mathrm{H}} \mathbf{H}_{\mathrm{S2V}}\mathbf{w}_{\mathrm{S}}^{(k-1)}\right|\\
\mathrm{Subject~ to}~~~~~ &\left | \mathbf{w}_{\mathrm{r}}^{\mathrm{H}}\mathbf{H}_{\mathrm{SI}}\mathbf{w}_{\mathrm{t}}^{(k-1)} \right | \leq \eta^{(k)}_{1}, \\
&\left|\left[ \mathbf{w}_{\mathrm{r}}\right]_{n}\right| \leq \frac{1}{\sqrt{N_{\mathrm{r}}^{\mathrm{tot}}}}, ~ 1 \leq n \leq N_{\mathrm{r}}^{\mathrm{tot}},
\end{aligned}
\end{equation}
where $\mathbf{w}_{\mathrm{S}}^{(k-1)}$ and $\mathbf{w}_{\mathrm{t}}^{(k-1)}$ are the fixed SN-BFV and Tx-BFV obtained in the $(k-1)$-th iteration, respectively, and $\eta^{(k)}_{1}$ is the interference suppression factor. The suppression factor successively decreases in each iteration. Besides, the CM constraint on the BFV is relaxed to a convex constraint in Problem \eqref{eq_problem_sub1}. We will show later that this relaxation has little influence on the performance.

Similarly, given the Rx-BFV obtained in Problem \eqref{eq_problem_sub1}, i.e., $\mathbf{w}_{\mathrm{r}}^{(k)}$, and the DN-BFV $\mathbf{w}_{\mathrm{D}}^{(k-1)}$, such that the received signal power of the S2V link and the interference from the S2D link are fixed, motivated by \eqref{eq_Rate_S2V}, \eqref{eq_Rate_V2D}, we optimize the Tx-BFV to maximize the received signal power of the V2D link, while suppressing the SI. Specifically, we solve the following problem:
\begin{equation}\label{eq_problem_sub2}
\begin{aligned}
\mathop{\mathrm{Maximize}}\limits_{\mathbf{w}_{\mathrm{t}}}~~~~~ &\left | \mathbf{w}_{\mathrm{D}}^{(k-1)\mathrm{H}}\mathbf{H}_{\mathrm{V2D}}\mathbf{w}_{\mathrm{t}} \right |\\
\mathrm{Subject~ to}~~~~~ &\left | \mathbf{w}_{\mathrm{r}}^{(k)\mathrm{H}}\mathbf{H}_{\mathrm{SI}}\mathbf{w}_{\mathrm{t}} \right | \leq \eta^{(k)}_{2}, \\
&\left|\left[ \mathbf{w}_{\mathrm{t}}\right]_{n}\right| \leq \frac{1}{\sqrt{N_{\mathrm{t}}^{\mathrm{tot}}}}, ~ 1 \leq n \leq N_{\mathrm{t}}^{\mathrm{tot}},
\end{aligned}
\end{equation}
where $\eta^{(k)}_{2}$ is the interference suppression factor.

After obtaining the Rx-BFV $\mathbf{w}_{\mathrm{r}}^{(k)}$ and the Tx-BFV $\mathbf{w}_{\mathrm{t}}^{(k)}$ in the $k$-th iteration, we optimize the SN-BFV and DN-BFV in a similar manner. Specifically, given the fixed DN-BFV $\mathbf{w}_{\mathrm{D}}^{(k-1)}$, we optimize the SN-BFV to maximize the received signal power of the S2V link, while suppressing the interference caused by the S2D link, i.e.,
\begin{equation}\label{eq_problem_sub3}
\begin{aligned}
\mathop{\mathrm{Maximize}}\limits_{\mathbf{w}_{\mathrm{S}}}~~~~~ &\left|\mathbf{w}_{\mathrm{r}}^{(k)\mathrm{H}} \mathbf{H}_{\mathrm{S2V}}\mathbf{w}_{\mathrm{S}}\right|\\
\mathrm{Subject~ to}~~~~~ &\left | \mathbf{w}_{\mathrm{D}}^{(k-1)\mathrm{H}}\mathbf{H}_{\mathrm{S2D}}\mathbf{w}_{\mathrm{S}} \right | \leq \eta^{(k)}_{3}, \\
&\left|\left[ \mathbf{w}_{\mathrm{S}}\right]_{n}\right| \leq \frac{1}{\sqrt{N_{\mathrm{S}}^{\mathrm{tot}}}}, ~ 1 \leq n \leq N_{\mathrm{S}}^{\mathrm{tot}},
\end{aligned}
\end{equation}

Finally, we optimize the DN-BFV to maximize the received signal power of the V2D link, while suppressing the interference caused by the S2D link, i.e.,
\begin{equation}\label{eq_problem_sub4}
\begin{aligned}
\mathop{\mathrm{Maximize}}\limits_{\mathbf{w}_{\mathrm{D}}}~~~~~ &\left | \mathbf{w}_{\mathrm{D}}^{\mathrm{H}}\mathbf{H}_{\mathrm{V2D}}\mathbf{w}_{\mathrm{t}}^{(k)} \right |\\
\mathrm{Subject~ to}~~~~~ &\left | \mathbf{w}_{\mathrm{D}}^{\mathrm{H}}\mathbf{H}_{\mathrm{S2D}}\mathbf{w}_{\mathrm{S}}^{(k)} \right | \leq \eta^{(k)}_{4}, \\
&\left|\left[ \mathbf{w}_{\mathrm{D}}\right]_{n}\right| \leq \frac{1}{\sqrt{N_{\mathrm{D}}^{\mathrm{tot}}}}, ~ 1 \leq n \leq N_{\mathrm{D}}^{\mathrm{tot}},
\end{aligned}
\end{equation}

To ensure that the interferences from the SI channel and the S2D channel are reduced in each iteration, we set $\eta^{(k)}_{i}=\eta+\mu^{(k)}_{i}$ for $i=\left\{1,2,3,4\right\}$, where $\eta$ is a nonnegative lower bound for the interference suppression factor. One possible choice is $\mu^{(k)}_{1}=\frac{\mu^{(k-1)}_{2}}{\kappa}$, $\mu^{(k)}_{2}=\frac{\mu^{(k)}_{1}}{\kappa}$, $\mu^{(k)}_{3}=\frac{\mu^{(k-1)}_{4}}{\kappa}$, and $\mu^{(k)}_{4}=\frac{\mu^{(k)}_{3}}{\kappa}$, where $\kappa$ is defined as the step size for the reduction of the interference suppression factor. The iterative process can be stopped when the increase of the achievable rate is no larger than a threshold $\epsilon_{r}$.

Problems \eqref{eq_problem_sub1}, \eqref{eq_problem_sub2}, \eqref{eq_problem_sub3}, and \eqref{eq_problem_sub4} have a similar form. Thus, we only develop the solution of Problem \eqref{eq_problem_sub1} in detail, and the other problems can be solved in the same manner. For Problem \eqref{eq_problem_sub1}, a convex objective function is maximized, which makes it a non-convex problem \cite{boyd2004convex}. Fortunately, a phase rotation of the BFVs does not impact the optimality of this problem. If $\mathbf{w}_{\mathrm{r}}^{\mathrm{\star}}$ is an optimal solution, then $\mathbf{w}_{\mathrm{r}}^{\mathrm{\star}}e^{j\pi \omega}$ is also an optimal solution. Exploiting this property, we can always find an optimal solution, where the argument of the magnitude operator $|\cdot|$ in the objective function of Problem \eqref{eq_problem_sub1} is a real number. Then, Problem \eqref{eq_problem_sub1} becomes equivalent to
\begin{equation}\label{eq_problem_sub1_eq}
\begin{aligned}
\mathop{\mathrm{Maximize}}\limits_{\mathbf{w}_{\mathrm{r}}}~~~~~ &\mathfrak{R}\left(\mathbf{w}_{\mathrm{r}}^{\mathrm{H}} \mathbf{H}_{\mathrm{S2V}}\mathbf{w}_{\mathrm{S}}^{(k-1)}\right)\\
\mathrm{Subject~ to}~~~~~ &\left | \mathbf{w}_{\mathrm{r}}^{\rm{H}}\mathbf{H}_{\mathrm{SI}}\mathbf{w}_{\mathrm{t}}^{(k-1)} \right | \leq \eta^{(k)}_{1},\\
&\left|\left[ \mathbf{w}_{\mathrm{r}}\right]_{n}\right| \leq \frac{1}{\sqrt{N_{\mathrm{r}}^{\mathrm{tot}}}}, ~ 1 \leq n \leq N_{\mathrm{r}}^{\mathrm{tot}},
\end{aligned}
\end{equation}
where $\mathfrak{R}(\cdot)$ denotes the real part of a complex number. Problem \eqref{eq_problem_sub1_eq} is a convex problem and can be solved by utilizing standard optimization tools such as CVX \cite{boyd2004convex}.

After obtaining the optimal solution of Problems \eqref{eq_problem_sub1}, \eqref{eq_problem_sub2}, \eqref{eq_problem_sub3}, and \eqref{eq_problem_sub4}, which we denote by $\mathbf{w}_{\mathrm{r}}^{\mathrm{\circ}}$, $\mathbf{w}_{\mathrm{t}}^{\mathrm{\circ}}$, $\mathbf{w}_{\mathrm{S}}^{\mathrm{\circ}}$, and $\mathbf{w}_{\mathrm{D}}^{\mathrm{\circ}}$, respectively, we normalize the modulus of the BFVs' elements to satisfy the CM constraint, i.e.,
\begin{equation}\label{eq_normal1}
\left[ \mathbf{w}_{\tau}^{(k)}\right]_{n} = \frac{1}{\sqrt{N_{\tau}^{\mathrm{tot}}}}\frac{\left[ \mathbf{w}_{\tau}^{\mathrm{\circ}}\right]_{n}}{\left|\left[ \mathbf{w}_{\tau}^{\mathrm{\circ}}\right]_{n}\right|}, ~ 1 \leq n \leq N_{\tau}^{\mathrm{tot}}, \tau=\left\{\mathrm{S}, \mathrm{r}, \mathrm{t}, \mathrm{D}\right\}.
\end{equation}

During the alternating optimization of the Tx-BFV and Rx-BFV in Problems \eqref{eq_problem_sub1} and \eqref{eq_problem_sub2}, respectively, the SI at the FD-UAV relay decreases successively, because the interference suppression factor decreases in each iteration. Similarly, the interference from the S2D link decreases successively, benefiting from the alternating optimization of the SN-BFV and DN-BFV in Problems \eqref{eq_problem_sub3} and \eqref{eq_problem_sub4}, respectively. Meanwhile, the beam gains of the target signals are maximized. With the AIS algorithm, the interference suppression factor finally converges to its lower bound $\eta$, and thus the powers of the SI and the interference from the S2D link are no larger than $\eta^2 P^{\mathrm{tot}}_{\mathrm{V}}$ and $\eta^2 P^{\mathrm{tot}}_{\mathrm{S}}$, respectively. To maximize the achievable rate, the interference powers should be restricted to be smaller than the noise powers, i.e., $\eta^2 P^{\mathrm{tot}}_{\mathrm{V}} < \sigma_{1}^2$ and $\eta^2 P^{\mathrm{tot}}_{\mathrm{S}}<\sigma_{2}^2$. Hence, a small $\eta$ is preferable to minimize the influence of the SI. However, a too small value of $\eta$ leads to smaller gains of the target signals because of the stricter interference constraints in \eqref{eq_problem_sub1}, \eqref{eq_problem_sub2}, \eqref{eq_problem_sub3}, and \eqref{eq_problem_sub4}. In fact, there is a tradeoff between the powers of the interferences and the powers of the target signals.

Now, the influence of the relaxation and normalization of the BFVs remains to be analyzed. To this end, we provide the following theorem.
\begin{theorem} \label{Theo_modulus}
There always exists an optimal solution of Problem \eqref{eq_problem_sub1}, where at most one element of the optimal BFV does not satisfy the CM constraint.
\end{theorem}

\begin{proof}
See Appendix B.
\end{proof}

Theorem \ref{Theo_modulus} suggests that the relaxation and normalization of the BFVs in \eqref{eq_normal1} have little influence on the rate performance because they impact at most one of their elements. In particular, when the number of antennas is large, the impact of a single element's normalization on the effective channel gain is small.

\subsection{Power Control}
As we have discussed before, to maximize the achievable rate from the SN to the DN and to avoid a waste of transmit power, the power control at the SN and FD-UAV relay should be carefully designed. Substituting the designed BFVs into \eqref{eq_Rate_S2V} and \eqref{eq_Rate_V2D}, we obtain the achievable rates of the S2V and V2D links as follows
\begin{equation}\label{eq_Rate_S2V3}
\tilde{R}_{\mathrm{S2V}}=\log_{2}\left (1+ \frac{G_{\mathrm{S2V}} P_{\mathrm{S}}}{G_{\mathrm{SI}} P_{\mathrm{V}}+\sigma_{1}^{2}}\right ),
\end{equation}
\begin{equation}\label{eq_Rate_V2D3}
\tilde{R}_{\mathrm{V2D}}=\log_{2}\left (1+ \frac{G_{\mathrm{V2D}} P_{\mathrm{V}}}{G_{\mathrm{S2D}}P_{\mathrm{S}}+\sigma_{2}^{2}}\right ),
\end{equation}
where $G_{\mathrm{S2V}}=\left | \mathbf{w}_{\mathrm{r}}^{(k)\rm{H}}\mathbf{H}_{\mathrm{S2V}}\mathbf{w}_{\mathrm{S}}^{(k)} \right |^{2}$, $G_{\mathrm{SI}}=\left | \mathbf{w}_{\mathrm{r}}^{(k)\mathrm{H}}\mathbf{H}_{\mathrm{SI}}\mathbf{w}_{\mathrm{t}}^{(k)} \right |^{2}$, $G_{\mathrm{V2D}}=\left | \mathbf{w}_{\mathrm{D}}^{(k)\mathrm{H}}\mathbf{H}_{\mathrm{V2D}}\mathbf{w}_{\mathrm{t}}^{(k)} \right |^{2}$, and $G_{\mathrm{S2D}}=\left | \mathbf{w}_{\mathrm{D}}^{(k)\rm{H}}\mathbf{H}_{\mathrm{S2D}}\mathbf{w}_{\mathrm{S}}^{(k)} \right |^{2}$.

To maximize the minimum of $\tilde{R}_{\mathrm{S2V}}$ and $\tilde{R}_{\mathrm{V2D}}$ as well as minimize the total transmit power, we provide the following theorem.

\begin{theorem} \label{Theo_power}
For given position and BFVs, the optimal power allocation for the SN and FD-UAV relay is given as follows

\begin{equation}\label{eq_opt_power} \footnotesize
\begin{aligned}
&\left\{
\begin{aligned}
&P_{\mathrm{S}}^{\star}=P_{\mathrm{S}}^{\mathrm{tot}},\\
&P_{\mathrm{V}}^{\star}=\frac{-b_{1}+\sqrt{b_{1}^2-4a_{1}c_{1}}}{2a_{1}},
\end{aligned}
\right.
~\text{if}~\frac{G_{\mathrm{S2V}} P_{\mathrm{S}}^{\mathrm{tot}}}{G_{\mathrm{SI}} P_{\mathrm{V}}^{\mathrm{tot}}+\sigma_{1}^{2}}<\frac{G_{\mathrm{V2D}} P_{\mathrm{V}}^{\mathrm{tot}}}{G_{\mathrm{S2D}}P_{\mathrm{S}}^{\mathrm{tot}}+\sigma_{2}^{2}};\\
&\left\{
\begin{aligned}
&P_{\mathrm{S}}^{\star}=\frac{-b_{2}+\sqrt{b_{2}^2-4a_{2}c_{2}}}{2a_{2}},\\
&P_{\mathrm{V}}^{\star}=P_{\mathrm{V}}^{\mathrm{tot}},
\end{aligned}
\right.
~\text{if}~\frac{G_{\mathrm{S2V}} P_{\mathrm{S}}^{\mathrm{tot}}}{G_{\mathrm{SI}} P_{\mathrm{V}}^{\mathrm{tot}}+\sigma_{1}^{2}} \geq \frac{G_{\mathrm{V2D}} P_{\mathrm{V}}^{\mathrm{tot}}}{G_{\mathrm{S2D}}P_{\mathrm{S}}^{\mathrm{tot}}+\sigma_{2}^{2}};
\end{aligned}
\end{equation}
where $a_{1}=G_{\mathrm{SI}}G_{\mathrm{V2D}}$, $b_{1}=G_{\mathrm{V2D}}\sigma_{1}^{2}$, $c_{1}=-G_{\mathrm{S2V}}P_{\mathrm{S}}^{\mathrm{tot}}\left(G_{\mathrm{S2D}}P_{\mathrm{S}}^{\mathrm{tot}}+\sigma_{2}^{2}\right)$, and $a_{2}=G_{\mathrm{S2D}}G_{\mathrm{S2V}}$, $b_{2}=G_{\mathrm{S2V}}\sigma_{2}^{2}$, $c_{2}=-G_{\mathrm{V2D}}P_{\mathrm{V}}^{\mathrm{tot}}\left(G_{\mathrm{SI}}P_{\mathrm{V}}^{\mathrm{tot}}+\sigma_{1}^{2}\right)$.
\end{theorem}

\begin{proof}
Note that our goal is to maximize the minimum of $\tilde{R}_{\mathrm{S2V}}$ and $\tilde{R}_{\mathrm{V2D}}$. Assume that the optimal transmit powers at the SN and the FD-UAV relay are both smaller than their maximum values, i.e., $P_{\mathrm{S}}^{\star}<P_{\mathrm{S}}^{\mathrm{tot}}$ and $P_{\mathrm{V}}^{\star}<P_{\mathrm{V}}^{\mathrm{tot}}$. We set $P_{\mathrm{S}}^{\circ}=(1+\delta)P_{\mathrm{S}}^{\star}$ and $P_{\mathrm{V}}^{\circ}=(1+\delta)P_{\mathrm{V}}^{\star}$, where $\delta$ is positive and small enough to ensure that $\left(P_{\mathrm{S}}^{\circ},P_{\mathrm{V}}^{\circ}\right)$ do not exceed the maximum values of the transmit powers. It can be verified that $\left(P_{\mathrm{S}}^{\circ},P_{\mathrm{V}}^{\circ}\right)$ yield a larger achievable rate than $\left(P_{\mathrm{S}}^{\star},P_{\mathrm{V}}^{\star}\right)$, which contradicts the assumption that $\left(P_{\mathrm{S}}^{\star},P_{\mathrm{V}}^{\star}\right)$ is optimal. Thus, we conclude that for the optimal power allocation, at least one of the transmit powers assumes the maximum possible value.

When $\frac{G_{\mathrm{S2V}} P_{\mathrm{S}}^{\mathrm{tot}}}{G_{\mathrm{SI}} P_{\mathrm{V}}^{\mathrm{tot}}+\sigma_{1}^{2}}<\frac{G_{\mathrm{V2D}} P_{\mathrm{V}}^{\mathrm{tot}}}{G_{\mathrm{S2D}}P_{\mathrm{S}}^{\mathrm{tot}}+\sigma_{2}^{2}}$, we have $\tilde{R}_{\mathrm{S2V}}<\tilde{R}_{\mathrm{V2D}}$. Thus, $P_{\mathrm{S}}^{\star}=P_{\mathrm{S}}^{\mathrm{tot}}$ maximizes the achievable rate of the S2V link. Meanwhile, to avoid the waste of transmit power and to maximize the achievable rate, $P_{\mathrm{V}}$ should be reduced, whereby the achievable rate of the V2D link decreases while the achievable rate of the S2V link increases. Solving equation ${\tilde{R}_{\mathrm{V2D}}=\tilde{R}_{\mathrm{S2V}}}$ for $P_{\mathrm{S}}^{\star}=P_{\mathrm{S}}^{\mathrm{tot}}$, we obtain the optimal transmit power of the FD-UAV relay as $P_{\mathrm{V}}^{\star}=\frac{-b_{1}+\sqrt{b_{1}^2-4a_{1}c_{1}}}{2a_{1}}$.

Similarly, when $\frac{G_{\mathrm{S2V}} P_{\mathrm{S}}^{\mathrm{tot}}}{G_{\mathrm{SI}} P_{\mathrm{V}}^{\mathrm{tot}}+\sigma_{1}^{2}} \geq \frac{G_{\mathrm{V2D}} P_{\mathrm{V}}^{\mathrm{tot}}}{G_{\mathrm{S2D}}P_{\mathrm{S}}^{\mathrm{tot}}+\sigma_{2}^{2}}$, we have $\tilde{R}_{\mathrm{S2V}}\geq\tilde{R}_{\mathrm{V2D}}$. Thus, $P_{\mathrm{V}}^{\star}=P_{\mathrm{V}}^{\mathrm{tot}}$ maximizes the achievable rate of the V2D link. Meanwhile, to avoid the waste of transmit power and to maximize the achievable rate, $P_{\mathrm{S}}$ should be reduced. Then, the achievable rate of the S2V link decreases while the achievable rate of the V2D link increases. Solving equation ${\tilde{R}_{\mathrm{V2D}}=\tilde{R}_{\mathrm{S2V}}}$ for $P_{\mathrm{V}}^{\star}=P_{\mathrm{V}}^{\mathrm{tot}}$, the optimal transmit power of the SN is obtained as $P_{\mathrm{S}}^{\star}=\frac{-b_{2}+\sqrt{b_{2}^2-4a_{2}c_{2}}}{2a_{2}}$. This concludes the proof.
\end{proof}

Hereto, we have obtained the optimal solution of the transmit power variables.

\subsection{Overall Solution}
We summarize the overall solution of the joint positioning, beamforming, and power control problem for mmWave FD-UAV relay systems in Algorithm 1. In line 1, we obtain the conditional optimal position of the FD-UAV relay based on Theorem \ref{Theo_posi}, assuming an LoS environment and ideal beamforming. In lines 2-11, we find the position of the FD-UAV relay in a neighborhood of the conditional optimal position. Then, in lines 17-31, we successively decrease the interferences by alternately solving Problems \eqref{eq_problem_sub1}, \eqref{eq_problem_sub2}, \eqref{eq_problem_sub3}, and \eqref{eq_problem_sub4}, where the optimal power allocation according to Theorem \ref{Theo_power} is incorporated in each iteration to maximize the achievable rate, see line 29. Note that the position of the FD-UAV relay is not updated during the iterative process as the obtained solution achieves a near-optimal performance if the proposed algorithm approaches ideal beamforming. The algorithm terminates if the improvement in the achievable rate from one iteration to the next falls below a threshold $\epsilon_{r}$. The convergence of Algorithm 1 will be studied via simulations in Section IV.

\begin{algorithm}[h] \small
\caption{Joint positioning, beamforming, and power control for mmWave FD-UAV relay systems.}
\begin{algorithmic}[1]
\REQUIRE ~$M_{\mathrm{S}}$, $N_{\mathrm{S}}$, $M_{\mathrm{D}}$, $N_{\mathrm{D}}$, $M_{\mathrm{t}}$, $N_{\mathrm{t}}$, $M_{\mathrm{r}}$, $N_{\mathrm{r}}$, $x_{\mathrm{D}}$, $y_{\mathrm{D}}$,\\ $h_{\min}$, $h_{\max}$ $P_{\mathrm{S}}^{\mathrm{tot}}$, $P_{\mathrm{V}}^{\mathrm{tot}}$, $\sigma_{1}$, $\sigma_{2}$, $f_{c}$ $\alpha_{\mathrm{LoS}}$, $\alpha_{\mathrm{NLoS}}$,\\ $\sigma_{f}$, $a$, $b$, $\epsilon_{x}$, $\epsilon_{y}$, $\epsilon_{h}$, $\eta$, $\kappa$, $\epsilon_{r}$.
\ENSURE ~${x_{\mathrm{V}}^{\mathrm{\circ}}, y_{\mathrm{V}}^{\mathrm{\circ}}, h_{\mathrm{V}}^{\mathrm{\circ}}, \mathbf{w}_{\mathrm{S}}^{\mathrm{\star}}, \mathbf{w}_{\mathrm{D}}^{\mathrm{\star}}, \mathbf{w}_{\mathrm{r}}^{\mathrm{\star}}, \mathbf{w}_{\mathrm{t}}^{\mathrm{\star}}}$, $P_{\mathrm{S}}^{\star}$, $P_{\mathrm{V}}^{\star}$. \\
\STATE Calculate $\left(x_{\mathrm{V}}^{\mathrm{\star}}, y_{\mathrm{V}}^{\mathrm{\star}}, h_{\mathrm{V}}^{\mathrm{\star}}\right)$ based on Theorem \ref{Theo_posi}.
\IF  {$\left(x_{\mathrm{V}}^{\mathrm{\star}}, y_{\mathrm{V}}^{\mathrm{\star}}, h_{\mathrm{V}}^{\mathrm{\star}}\right)$ has an LoS environment}
\STATE Set $\left(x_{\mathrm{V}}^{\mathrm{\circ}}, y_{\mathrm{V}}^{\mathrm{\circ}}, h_{\mathrm{V}}^{\mathrm{\circ}}\right)=\left(x_{\mathrm{V}}^{\mathrm{\star}}, y_{\mathrm{V}}^{\mathrm{\star}}, h_{\mathrm{V}}^{\mathrm{\star}}\right)$.
\ELSE
\STATE Initialize $t=0$, $\mathcal{C}_{0}=\left\{\left(x_{\mathrm{V}}^{\mathrm{\star}}, y_{\mathrm{V}}^{\mathrm{\star}}, h_{\mathrm{V}}^{\mathrm{\star}}\right)\right\}$, $\mathcal{L}_{0}={\O}$.
\WHILE {$\mathcal{L}_{t}$ is empty}
\STATE Update $t=t+1$.
\STATE Obtain $\mathcal{C}_{t}$ and $\mathcal{L}_{t}$.
\ENDWHILE
\STATE Determine $\left(x_{\mathrm{V}}^{\mathrm{\circ}}, y_{\mathrm{V}}^{\mathrm{\circ}}, h_{\mathrm{V}}^{\mathrm{\circ}}\right)$ based on \eqref{AdjPosition}.
\ENDIF
\STATE Estimate channel matrices $\mathbf{H}_{\mathrm{S2V}}$, $\mathbf{H}_{\mathrm{V2D}}$, $\mathbf{H}_{\mathrm{S2D}}$, and $\mathbf{H}_{\mathrm{SI}}$.
\STATE Initialize $k=0$.
\STATE Initialize $\mathbf{w}_{\mathrm{S}}^{(0)}$, $\mathbf{w}_{\mathrm{D}}^{(0)}$, $\mathbf{w}_{\mathrm{r}}^{(0)}$ and $\mathbf{w}_{\mathrm{t}}^{(0)}$ according to \eqref{eq_BFV_SN}.
\STATE Initialize $\mu^{(0)}_{2}=\left | \mathbf{w}_{\mathrm{r}}^{\rm{(0)H}}\mathbf{H}_{\mathrm{SI}}\mathbf{w}_{\mathrm{t}}^{(0)} \right |$.
\STATE Calculate $R_{\mathrm{S2D}}^{(0)}$ according to \eqref{eq_Rate_S2D} and define $R_{\mathrm{S2D}}^{(-1)}=-\infty$.
\WHILE {$R_{\mathrm{S2D}}^{(k)}-R_{\mathrm{S2D}}^{(k-1)}>\epsilon_{r}$}
\STATE $k=k+1$.
\STATE Update the suppression factor $\mu^{(k)}_{i}=\frac{\mu^{(k-1)}_{i+1}}{\kappa}$ and $\eta^{(k)}_{i} = \eta+ \mu^{(k)}_{i}$ for $i=1,3$.
\STATE Update the suppression factor $\mu^{(k)}_{i}=\frac{\mu^{(k)}_{i-1}}{\kappa}$ and $\eta^{(k)}_{i} = \eta+ \mu^{(k)}_{i}$ $i=2,4$.
\STATE Solve Problem \eqref{eq_problem_sub1} to obtain $\mathbf{w}_{\mathrm{r}}^{\circ}$.
\STATE Normalize $\mathbf{w}_{\mathrm{r}}^{\circ}$ according to \eqref{eq_normal1} and obtain $\mathbf{w}_{\mathrm{r}}^{(k)}$.
\STATE Solve Problem \eqref{eq_problem_sub2} to obtain $\mathbf{w}_{\mathrm{t}}^{\circ}$.
\STATE Normalize $\mathbf{w}_{\mathrm{t}}^{\circ}$ according to \eqref{eq_normal1} and obtain $\mathbf{w}_{\mathrm{t}}^{(k)}$.
\STATE Solve Problem \eqref{eq_problem_sub3} to obtain $\mathbf{w}_{\mathrm{S}}^{\circ}$.
\STATE Normalize $\mathbf{w}_{\mathrm{S}}^{\circ}$ according to \eqref{eq_normal1} and obtain $\mathbf{w}_{\mathrm{S}}^{(k)}$.
\STATE Solve Problem \eqref{eq_problem_sub4} to obtain $\mathbf{w}_{\mathrm{D}}^{\circ}$.
\STATE Normalize $\mathbf{w}_{\mathrm{D}}^{\circ}$ according to \eqref{eq_normal1} and obtain $\mathbf{w}_{\mathrm{D}}^{(k)}$.
\STATE Obtain $P_{\mathrm{S}}^{(k)}$ and $P_{\mathrm{V}}^{(k)}$ according to Theorem \ref{Theo_power}.
\STATE Calculate $R_{\mathrm{S2D}}^{(k)}$ according to \eqref{eq_Rate_S2D}.
\ENDWHILE
\STATE $\mathbf{w}_{\mathrm{r}}^{\mathrm{\star}}=\mathbf{w}_{\mathrm{r}}^{(k)} $, $\mathbf{w}_{\mathrm{t}}^{\mathrm{\star}}=\mathbf{w}_{\mathrm{t}}^{(k)}$, $P_{\mathrm{S}}^{\star}=P_{\mathrm{S}}^{(k)}$, and $P_{\mathrm{V}}^{\star}=P_{\mathrm{V}}^{(k)}$.
\RETURN ${x_{\mathrm{V}}^{\mathrm{\star}}, y_{\mathrm{V}}^{\mathrm{\star}}, h_{\mathrm{V}}^{\mathrm{\star}}, \mathbf{w}_{\mathrm{S}}^{\mathrm{\star}}, \mathbf{w}_{\mathrm{D}}^{\mathrm{\star}}, \mathbf{w}_{\mathrm{r}}^{\mathrm{\star}}, \mathbf{w}_{\mathrm{t}}^{\mathrm{\star}}}$, $P_{\mathrm{S}}^{\star}$, $P_{\mathrm{V}}^{\star}$.
\end{algorithmic}
\end{algorithm}

In the proposed joint positioning, beamforming, and power control algorithm, the FD-UAV positioning is determined first and entails a maximum computational complexity of $\mathcal{O}\left(K_{x}K_{y}K_{h}\right)$, where $K_{x}=\lceil\frac{x_{\mathrm{D}}}{\epsilon_{x}}\rceil$, $K_{y}=\lceil\frac{y_{\mathrm{D}}}{\epsilon_{y}}\rceil$, and $K_{h}=\lceil\frac{h_{\max}-h_{\min}}{\epsilon_{h}}\rceil$ are the maximum possible numbers of candidate coordinates for directions $x$, $y$, and $z$, respectively. The complexity of solving Problem \eqref{eq_problem_sub1} by using the interior point method and the normalization of the Rx-BFV is $\mathcal{O}\left({N_{\mathrm{r}}^{\mathrm{tot}}}^{3.5}\right)$ and $\mathcal{O}\left(N_{\mathrm{r}}^{\mathrm{tot}}\right)$, respectively \cite{boyd2004convex}. Then, the complexity of the joint beamforming and power control process from line 19 to 30 in Algorithm 1 is $\mathcal{O}\left({N_{\max}^{\mathrm{tot}}}^{3.5}\right)$, where $N_{\max}^{\mathrm{tot}}=\max \{{N_{\mathrm{r}}^{\mathrm{tot}}},{N_{\mathrm{t}}^{\mathrm{tot}}},{N_{\mathrm{S}}^{\mathrm{tot}}},
{N_{\mathrm{D}}^{\mathrm{tot}}}\}$. As a result, the overall computational complexity of Algorithm 1 is $\mathcal{O}\left(K_{x}K_{y}K_{h}+T{N_{\max}^{\mathrm{tot}}}^{3.5}\right)$, where $T$ is the maximum number of iterations of the AIS algorithm.

\section{Performance Evaluation}
In this section, we provide simulation results to evaluate the performance of the proposed joint positioning, beamforming, and power control scheme for mmWave FD-UAV relay systems.

\subsection{Simulation Setup and Benchmark Schemes}
We adopt the channel models in \eqref{Channel_S2V}, \eqref{Channel_V2D}, \eqref{Channel_S2D}, and \eqref{Channel_SI}, where the probabilities that an LoS path exists for the S2V and V2D channels are given by \eqref{prob_LoS_S2V} and \eqref{prob_LoS_V2D}, respectively. The number of NLoS components for the S2V, V2D, and S2D channels are assumed to be identical, i.e., $L_{\mathrm{S2V}}=L_{\mathrm{V2D}}=L_{\mathrm{S2D}}=L$. The adopted simulation parameter settings are provided in Table \ref{tab:para} \cite{Hourani2014MaxCov,Rappaport2015meas}, unless specified otherwise. Half-wavelength spacing UPAs are used at all nodes, and the Tx-UPA and Rx-UPA at the FD-UAV relay are parallel to each other with a distance of $10\lambda$ ($\approx$ 8 cm). For the proposed AIS algorithm, the lower bound for the SI suppression factor is set to $\eta=\min\left\{\frac{\sigma_{1}}{10\sqrt{P_{\mathrm{S}}^{\mathrm{tot}}}},\frac{\sigma_{2}}{10\sqrt{P_{\mathrm{V}}^{\mathrm{tot}}}}\right\}$, such that the interference power is in the same range as the noise power. Each simulation point is averaged over $10^3$ node distributions and channel realizations, where the DN is randomly distributed in a disk of radius 500 m, with the SN at its center.

\begin{table}[h]
\caption{Simulation Parameters}\label{tab:para}
\footnotesize
\begin{center}
\begin{tabular}{|c|c|c|}
  \hline
 \textbf{Parameter}                        & \textbf{Description}                                        & \textbf{Value} \\
  \hline
  $h_{\min}$                                 & Minimum altitude of UAV                                   & 100 m \\
  \hline
  $h_{\max}$                                 & Maximum altitude of UAV                                   & 300 m \\
  \hline
  $P_{\mathrm{S}}^{\mathrm{tot}}$            & Maximum transmit power of the SN                          & 20 dBm \\
  \hline
  $P_{\mathrm{V}}^{\mathrm{tot}}$            & Maximum transmit power of the UAV                         & 20 dBm \\
  \hline
  $\sigma_{1}^{2}$                           & Power of the noise at the UAV                             & -110 dBm \\
  \hline
  $\sigma_{2}^{2}$                           & Power of the noise at the DN                              & -110 dBm \\
  \hline
  $f_c$ ($=c/\lambda$)                       & Carrier frequency                                         & 38 GHz \\
  \hline
  $\alpha_{\mathrm{LoS}}$                    & Path loss exponent for LoS paths                          & 1.9 \\
  \hline
  $\alpha_{\mathrm{NLoS}}$                   & Path loss exponent for NLoS paths                         & 3.3 \\
  \hline
  $L$                                        & Number of NLoS components                                 & 4 \\
  \hline
  $\sigma_{f}$                               & Standard deviation of shadow factor                       & $1/\sqrt{L}$ \\
  \hline
  $a$                                        & Environment parameter in \eqref{prob_LoS_S2V} and \eqref{prob_LoS_V2D}              & 11.95 \\
  \hline
  $b$                                        & Environment parameter in \eqref{prob_LoS_S2V} and \eqref{prob_LoS_V2D}              & 0.14 \\
  \hline
  $M_{\mathrm{S}} \times N_{\mathrm{S}}$     & Antenna array size at the SN                              & $4 \times 4$\\
  \hline
  $M_{\mathrm{D}} \times N_{\mathrm{D}}$     & Antenna array size at the DN                              & $4 \times 4$\\
  \hline
  $M_{\mathrm{t}} \times N_{\mathrm{t}}$     & Antenna array size of Tx-UPA at the UAV                   & $4 \times 4$\\
  \hline
  $M_{\mathrm{r}} \times N_{\mathrm{r}}$     & Antenna array size of Rx-UPA at the UAV                   & $4 \times 4$\\
  \hline
  $\epsilon_{x}$                             & Granularity for coordinate $x$                              & 1 m \\
  \hline
  $\epsilon_{y}$                             & Granularity for coordinate $y$                              & 1 m \\
  \hline
  $\epsilon_{h}$                             & Granularity for coordinate $h$                              & 1 m \\
  \hline
  $\kappa$                                   & Step size for AIS                                         & 10\\
  \hline
  $\epsilon_{r}$                                 & Threshold for convergence of Algorithm 1                  & 0.01 bps/Hz\\
  \hline
\end{tabular}
\end{center}
\end{table}

Two upper bounds for the achievable rate for mmWave FD-UAV relay systems are considered. The proposed approximate upper bound is obtained as the minimum of \eqref{eq_RateBound_S2V} and \eqref{eq_RateBound_V2D}, while the proposed strict upper bound is the minimum of \eqref{eq_RateBound2_S2V} and \eqref{eq_RateBound2_V2D}. For both upper bounds, the FD-UAV relay is assumed to be at the designed position $(x_{\mathrm{V}}^{\mathrm{\circ}}, y_{\mathrm{V}}^{\mathrm{\circ}}, h_{\mathrm{V}}^{\mathrm{\circ}})$. Furthermore, three benchmark schemes are used for comparison, namely ``RandPos \& AIS'', ``DesPos \& steer'', and ``DesPos \& OMP'', respectively. For the ``RandPos \& AIS'' scheme, the position of the FD-UAV relay is randomly selected from the feasible region of Problem \eqref{eq_problem}, and the proposed AIS algorithm is employed for beamforming. For the ``DesPos \& steer'' scheme, the designed position for the FD-UAV relay, i.e., $(x_{\mathrm{V}}^{\mathrm{\circ}}, y_{\mathrm{V}}^{\mathrm{\circ}}, h_{\mathrm{V}}^{\mathrm{\circ}})$ given by \eqref{AdjPosition}, is employed, and the steering vectors in \eqref{eq_BFV_SN} are used for beamforming. For the ``DesPos \& OMP'' scheme, the designed position for the FD-UAV relay is employed, and the BFVs are obtained by utilizing the OMP-based SI-cancellation precoding algorithm in \cite{zhang2019EEFD}, where the number of RF chains is denoted by $N_{\mathrm{RF}}$. For all benchmark schemes, the optimal transmit powers from Theorem \ref{Theo_power} are adopted at the SN and the FD-UAV relay.

\subsection{Simulation Results}
\begin{figure}[t]
\begin{center}
  \includegraphics[width=\figwidth cm]{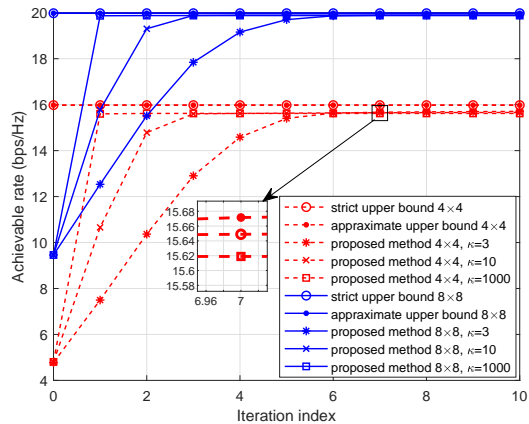}
  \caption{Evaluation of the convergence of the proposed Algorithm 1 for different values of $\kappa$.}
  \label{Fig:iteration_ka}
\end{center}
\end{figure}
\begin{figure}[t]
\begin{center}
  \includegraphics[width=\figwidth cm]{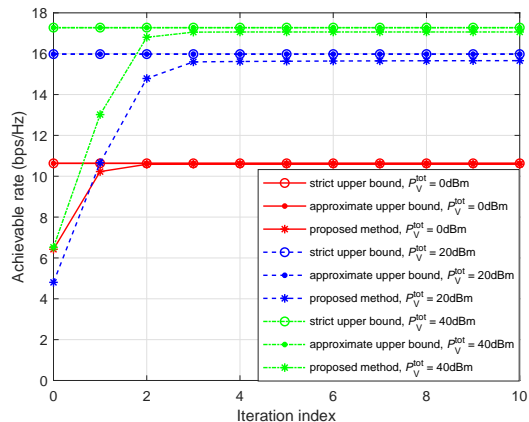}
  \caption{Evaluation of the convergence of the proposed Algorithm 1 for different maximum transmit powers of the FD-UAV relay.}
  \label{Fig:iteration_PUAV}
\end{center}
\end{figure}

First, in Fig. \ref{Fig:iteration_ka}, we evaluate the convergence of the proposed AIS beamforming method (Algorithm 1) for different step sizes for the reduction of the interference suppression factor (i.e., $\kappa$ in Algorithm 1). Identical sizes are adopted for the UPA at the SN, the UPA at the DN, and the Tx and Rx UPAs at the FD-UAV relay, i.e., $4\times4$ or $8\times8$. As can be observed, the proposed ASIS beamforming method converges very fast to a value close to the performance upper bound, and the approximate upper bound is very close to the strict upper bound. These results confirm the assumption of a pure LoS environment in Section III-A because the LoS path has much higher power compared to the NLoS paths. When the antenna array size is $4 \times 4$ at the FD-UAV relay, after convergence, the performance gap between the proposed method and the upper bound is no more than 0.3 bps/Hz, and this gap reduces to 0.1 bps/Hz when the antenna array size is $8 \times 8$. For larger numbers of antennas, there are more DoFs for minimization of the SI. Thus, the performance gap between the proposed method and the upper bound becomes smaller. The results in Fig. \ref{Fig:iteration_ka} demonstrate that the proposed method can achieve a near-upper-bound performance in terms of the achievable rate. In addition, the speed of convergence of the proposed AIS algorithm depends on the step size for the reduction of the suppression factor. For larger $\kappa$, the AIS algorithm converges faster. However, if $\kappa$ is chosen too large, for example, $\kappa \rightarrow +\infty$, the SI decreases too fast in the first iteration for designing $\mathbf{w}_{\mathrm{r}}^{(k)}$. As such, the effective channel gain of the S2V link may be much smaller than that of the V2D link, which negatively affects the achievable rate of the DN. Thus, to achieve a favorable tradeoff between the achievable rate and computational complexity, we set $\kappa =10$ for the following simulations.

Fig. \ref{Fig:iteration_PUAV} shows the convergence performance of the proposed AIS algorithm for different maximum transmit powers of the FD-UAV relay. For all considered cases, the proposed algorithm converges to a near-upper-bound achievable rate within few iterations, where all curves reach steady state after 4 iterations. Particularly, as the maximum transmit power at the FD-UAV relay increases, the number of the iterations required for convergence increases. The reason is that a higher transmit power of the UAV causes more SI, and thus more iterations are required to successively reduce the SI.

\begin{figure}[t]
\begin{center}
  \includegraphics[width=\figwidth cm]{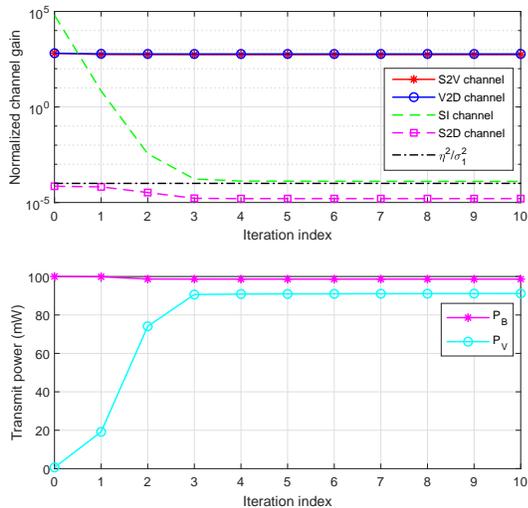}
  \caption{Normalized channel gains and transmit powers versus iteration index.}
  \label{Fig:Gain_Power}
\end{center}
\end{figure}

To shed more light on the properties of Algorithm 1, in Fig. \ref{Fig:Gain_Power}, we show the change of the channel gains and transmit powers during the iterations. In particular, we show the normalized channel gains, which are the ratios of the effective channel gains and the noise power in \eqref{eq_Rate_S2V} and \eqref{eq_Rate_V2D}, i.e., $\left | \mathbf{w}_{\mathrm{r}}^{\rm{H}}\mathbf{H}_{\mathrm{S2V}}\mathbf{w}_{\mathrm{S}} \right |^{2}/\sigma_{1}^{2}$, $\left | \mathbf{w}_{\mathrm{r}}^{\rm{H}}\mathbf{H}_{\mathrm{SI}}\mathbf{w}_{\mathrm{t}} \right |^{2}/\sigma_{1}^{2}$, $\left | \mathbf{w}_{\mathrm{D}}^{\rm{H}}\mathbf{H}_{\mathrm{V2D}}\mathbf{w}_{\mathrm{t}} \right |^{2}/\sigma_{1}^{2}$, and $\left | \mathbf{w}_{\mathrm{D}}^{\rm{H}}\mathbf{H}_{\mathrm{S2D}}\mathbf{w}_{\mathrm{S}} \right |^{2}/\sigma_{1}^{2}$. As can be observed, the channel gain of the SI channel decreases fast and converges to the lower bound $\eta^2/\sigma_{1}^{2}$, since the SI suppression factor is reduced in each iteration in \eqref{eq_problem_sub1} and \eqref{eq_problem_sub2}. The channel gain of the S2D channel is always lower than that of the SI channel because of the long transmission distance and the blockage of the LoS link between SN and DN. Besides, the channel gains of the S2V and V2D links remain almost unchanged during the iterations, which confirms the rational behind the proposed AIS beamforming algorithm. This is also the reason for why the achievable rate of the proposed scheme can approach the performance upper bound. For the variation of transmit powers, during the first iteration, the transmit power of the FD-UAV relay is very low, while the SN transmits with the maximal power. This is because the S2V link suffers from high SI for the initially chosen BFVs, and thus the FD-UAV reduces the transmit power to decrease the SI. After several iterations, the effective channel gain of the SI channel becomes lower, and thus the FD-UAV relay can increase its transmit power to improve the achievable rate of the V2D link.

\begin{figure}[t]
\begin{center}
  \includegraphics[width=\figwidth cm]{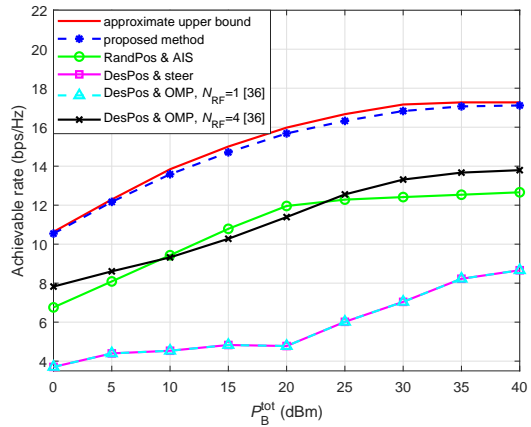}
  \caption{Achievable rates of different methods versus SN transmit powers.}
  \label{Fig:rate_PBS}
\end{center}
\end{figure}
\begin{figure}[t]
\begin{center}
  \includegraphics[width=\figwidth cm]{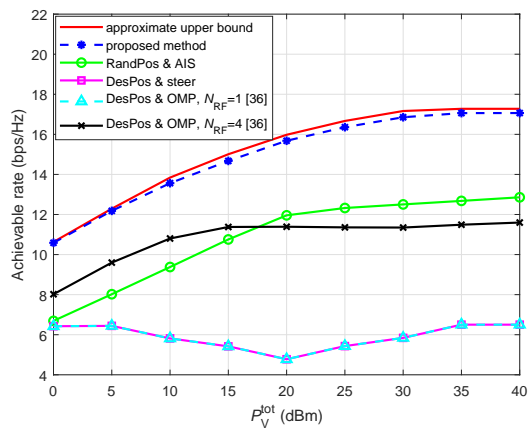}
  \caption{Achievable rates of different methods versus FD-UAV relay transmit powers.}
  \label{Fig:rate_PUAV}
\end{center}
\end{figure}

Fig. \ref{Fig:rate_PBS} compares the achievable rate performance of different methods as a function of the SN transmit power. As can be observed, the proposed joint position, beamforming, and power control method achieves a performance very close to the performance upper bound, and outperforms all benchmark schemes. In addition, as $P^{\mathrm{tot}}_{\mathrm{S}}$ increases, the speed of the increase of the achievable rate becomes smaller. The reason for this behavior is as follows. According to Theorem \ref{Theo_posi}, the conditional optimal position of the FD-UAV relay moves towards the DN as the transmit power of the SN increases. When $P^{\mathrm{tot}}_{\mathrm{S}}$ is sufficiently large, the conditional optimal position of the FD-UAV relay is right above the DN, and the achievable rate of the V2D link cannot increase anymore. In other words, the overall achievable rate is limited by the rate of the V2D link. We also observe that for one RF chain, the OMP-based SI-cancellation precoding algorithm in \cite{zhang2019EEFD} yields a similar performance as the steering vector-based beamforming scheme. When the number of RF chains increases, more SI can be mitigated in the digital beamforming domain, and the performance of the ``DesPos \& OMP'' scheme improves \cite{zhang2019EEFD}.

Fig. \ref{Fig:rate_PUAV} compares the achievable rate performance of different methods as a function of the FD-UAV relay transmit power. The proposed scheme outperforms again all benchmark schemes. As $P^{\mathrm{tot}}_{\mathrm{V}}$ increases, the achievable rate of the proposed method improves, but the rate of improvement decreases. The reason for this is that the position of the FD-UAV relay moves towards the SN as $P_{\mathrm{V}}$ increases. When the transmit power of the FD-UAV relay is sufficiently large, the conditional optimal position of the FD-UAV relay is right above the SN, and the achievable rate of the S2V link cannot be further improved and limits the overall performance. In addition, as the transmit power of the FD-UAV relay increases, the achievable rate of the ``DesPos \& steer'' scheme remains low because the SI is high at the FD-UAV relay if the steering vectors are employed for beamforming. The results in Figs. \ref{Fig:rate_PBS} and \ref{Fig:rate_PUAV} indicate that both the UAV positioning and the BFVs have a significant impact on the achievable-rate performance of mmWave FD-UAV relay systems.

\begin{figure}[t]
\begin{center}
  \includegraphics[width=\figwidth cm]{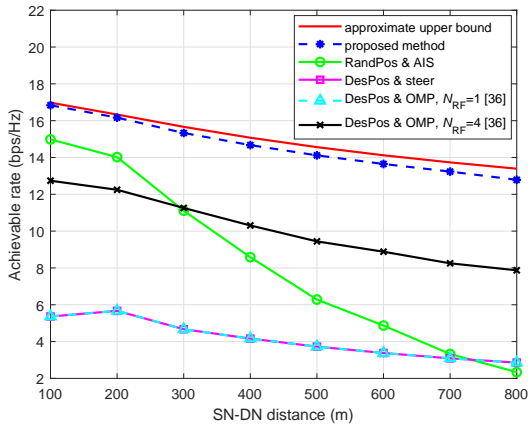}
  \caption{Achievable rates of different methods versus SN-DN distance.}
  \label{Fig:rate_D}
\end{center}
\end{figure}
\begin{figure}[t]
\begin{center}
  \includegraphics[width=\figwidth cm]{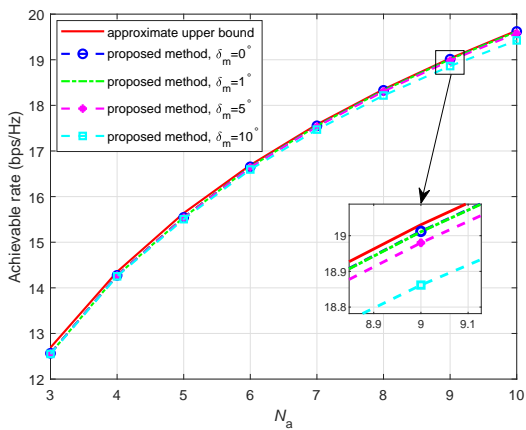}
  \caption{Achievable rates of different methods versus antenna array sizes for $M_{\tau}=N_{\tau}=N_{\text{a}}$ and $\tau=\{\mathrm{t},\mathrm{r},\mathrm{S},\mathrm{D}\}$.}
  \label{Fig:rate_MN}
\end{center}
\end{figure}

Fig. \ref{Fig:rate_D} compares the achievable rate of different methods as a function of the SN-DN distance. For each point on the horizontal axis, the DN is randomly distributed on a circle with the SN at its center and a fixed radius, i.e., the SN-DN distance. As can be observed, the achievable rates for the five considered schemes all decrease as the distance increases because the path loss increases. In particular for the ``RandPos \& AIS'' scheme, the achievable rate decreases rapidly with increasing distance. The reason for this behaviour is that, for larger SN-DN distances, the range of possible UAV positions increases, and the randomly deployed UAV may be further from the conditional optimal position.

Fig. \ref{Fig:rate_MN} compares the achievable rate of different methods as a function of the antenna array size for $M_{\tau}=N_{\tau}=N_{\text{a}}$ and $\tau=\{\mathrm{t},\mathrm{r},\mathrm{S},\mathrm{D}\}$. As the antenna array size increases, the achievable rate of the proposed joint positioning, beamforming, and power control method also increases because higher array gains can be obtained and more DoFs are available for suppression of the SI. However, due to the jitter of the UAV, the elevation angles and the azimuth angles of the air-to-ground channels may change rapidly, which results in beam misalignment. To evaluate the impact of beam misalignment, we model the real AoDs/AoAs of the S2V link and the V2D link as uniformly distributed random variables with fixed means and deviation $\delta_{\mathrm{m}}$, i.e., $\bar{\theta}_{\tau}^{(\ell)} \in \left[\theta_{\tau}^{(\ell)}-\delta_{\mathrm{m}}/2,\theta_{\tau}^{(\ell)}+\delta_{\mathrm{m}}/2\right]$ and $\bar{\phi}_{\tau}^{(\ell)} \in \left[\phi_{\tau}^{(\ell)}-\delta_{\mathrm{m}}/2,\phi_{\tau}^{(\ell)}+\delta_{\mathrm{m}}/2\right]$ for $\tau=\{\mathrm{t},\mathrm{r},\mathrm{S},\mathrm{D}\}$. The BFVs and power control are designed based on the estimated AoDs and AoAs ($\theta_{\tau}^{(\ell)}$ and $\phi_{\tau}^{(\ell)}$), while the achievable rates are calculated based on the real AoDs and AoAs ($\bar{\theta}_{\tau}^{(\ell)}$ and $\bar{\phi}_{\tau}^{(\ell)}$). As can be observed from Fig. \ref{Fig:rate_MN}, the achievable rates are very close to the upper bound for $\delta_{\mathrm{m}}=1^{\circ}$, $\delta_{\mathrm{m}}=5^{\circ}$, and $\delta_{\mathrm{m}}=10^{\circ}$. The reason is as follows. According to the array theory, the half-power beamwidth for a linear phased array employing steering vectors is $\Theta=2\left|\theta_{m}-\theta_{h}\right|$, where $\theta_{m}=\cos^{-1}\left(\frac{\beta \lambda}{2\pi d}\right)$ is the angle maximizing the array gain, $\theta_{h}=\cos^{-1}\left[\frac{\lambda}{2\pi d}(-\beta \pm \frac{2.782}{N})\right]$ is the 3-dB point for the array gain, $\beta$ is the difference in phase excitation between the antenna elements, and $N$ is the array size \cite{balanis2016antenna}. For $N=9$, $\beta=0$, and $d/\lambda=1/2$, the half-power beamwidth is $\Theta \approx 11.3^{\circ}$. Thus, beam misalignments with deviations not exceeding $10^{\circ}$ have little impact on the achievable rate. For larger array sizes, the beamwidth decreases and the impact of beam misalignment becomes more significant. The results in Fig. \ref{Fig:rate_MN} demonstrate the robustness of the proposed AIS beamforming algorithm with respect to beam misalignment.

\section{Conclusion}
In this paper, we proposed to employ an FD-UAV relay to improve the achievable rate of a mmWave communication system, where the SN, DN, and FD-UAV relay are all equipped with UPAs and use directional beams to overcome the high path loss of mmWave signals. Analog beamforming was utilized to mitigate the SI at the FD-UAV relay. We formulated a joint optimization problem for the UAV positioning, analog beamforming, and power control for maximization of the minimum of the achievable rates of the S2V and V2D links. To solve this highly non-convex, highly coupled, and high-dimensional problem, we first obtained the conditional optimal position of the FD-UAV relay for maximization of an approximate upper bound for the achievable rate, under the assumption of an LoS environment and ideal beamforming. Then, the UAV was deployed at the position which was closest to the conditional optimal position and yielded LoS paths for both the S2V and the V2D links. Subsequently, we developed an iterative algorithm for joint optimization of the BFVs and the power control variables. In each iteration, the BFVs were optimized  for maximization of the beam gains of the target signals and successive reduction of the interference, and the optimal power control variables were updated in closed form. Simulation results demonstrated that the proposed joint positioning, beamforming, and power control method for mmWave FD-UAV relay system can closely approach a performance upper bound in terms of the achievable rate and significantly outperforms three benchmark schemes.

\appendices
\section{Proof of Theorem \ref{Theo_posi}}
Based on \eqref{eq_RateBound_S2V} and \eqref{eq_RateBound_V2D}, we find that to maximize the achievable rate, the FD-UAV relay should always be deployed on the line segment between the SN and the DN with the minimum altitude. Otherwise, the S2V and V2D distances would both increase, which results in an additional propagation loss. Thus, we can set the coordinates of the UAV as $\left(x_{\mathrm{V}},y_{\mathrm{V}},h_{\mathrm{V}}\right)=\left(\rho x_{\mathrm{D}},\rho y_{\mathrm{D}}, h_{\min}\right)$, where $0\leq \rho \leq 1$.

Notice that the objective in Problem \eqref{eq_problem} is to maximize the minimal rate of the S2V and V2D links. If $\frac{N_{\mathrm{S}}^{\mathrm{tot}} N_{\mathrm{r}}^{\mathrm{tot}} P_{\mathrm{S}}^{\mathrm{tot}}\sigma_{2}^{2}}{N_{\mathrm{t}}^{\mathrm{tot}} N_{\mathrm{D}}^{\mathrm{tot}} P_{\mathrm{V}}^{\mathrm{tot}}\sigma_{1}^{2}} \leq \frac{h_{\min}^{\alpha_{\mathrm{LoS}}}}{\left(x_{\mathrm{D}}^{2}+y_{\mathrm{D}}^{2}+h_{\min}^{2}\right)^{\frac{\alpha_{\mathrm{LoS}}}{2}}}$, we have $\bar{R}_{\mathrm{S2V}} \leq \bar{R}_{\mathrm{V2D}}$ for any $\rho \in [0,1]$. Thus, the FD-UAV relay should be deployed right at the SN to maximize the minimal rate, i.e., $\bar{R}_{\mathrm{S2V}}$. As a result, the optimal coordinates of the UAV are obtained for $\rho^{\star}=0$.

Similarly, if $\frac{N_{\mathrm{S}}^{\mathrm{tot}} N_{\mathrm{r}}^{\mathrm{tot}} P_{\mathrm{S}}^{\mathrm{tot}}\sigma_{2}^{2}}{N_{\mathrm{t}}^{\mathrm{tot}} N_{\mathrm{D}}^{\mathrm{tot}} P_{\mathrm{V}}^{\mathrm{tot}}\sigma_{1}^{2}} \geq \frac{\left(x_{\mathrm{D}}^{2}+y_{\mathrm{D}}^{2}+h_{\min}^{2}\right)^{\frac{\alpha_{\mathrm{LoS}}}{2}}}{h_{\min}^{\alpha_{\mathrm{LoS}}}}$, we have $\bar{R}_{\mathrm{S2V}} \geq \bar{R}_{\mathrm{V2D}}$ for any $\rho \in [0,1]$. Thus, the FD-UAV relay should be deployed right at the DN to maximize the minimal rate, i.e., $\bar{R}_{\mathrm{V2D}}$. As a result, the optimal coordinates of the UAV are obtained for $\rho^{\star}=1$.

For the case $ \frac{h_{\min}^{\alpha_{\mathrm{LoS}}}}{\left(x_{\mathrm{D}}^{2}+y_{\mathrm{D}}^{2}+h_{\min}^{2}\right)^{\frac{\alpha_{\mathrm{LoS}}}{2}}}< \frac{N_{\mathrm{S}}^{\mathrm{tot}} N_{\mathrm{r}}^{\mathrm{tot}} P_{\mathrm{S}}^{\mathrm{tot}}\sigma_{2}^{2}}{N_{\mathrm{t}}^{\mathrm{tot}} N_{\mathrm{D}}^{\mathrm{tot}} P_{\mathrm{V}}^{\mathrm{tot}}\sigma_{1}^{2}} <\frac{\left(x_{\mathrm{D}}^{2}+y_{\mathrm{D}}^{2}+h_{\min}^{2}\right)^{\frac{\alpha_{\mathrm{LoS}}}{2}}}{h_{\min}^{\alpha_{\mathrm{LoS}}}} $, the relative size of $\bar{R}_{\mathrm{S2V}}$ and $\bar{R}_{\mathrm{V2D}}$ depends on the value of $\rho$. It is easy to verify that $\bar{R}_{\mathrm{S2V}}$ is decreasing in $\rho$, while $\bar{R}_{\mathrm{V2D}}$ is increasing in $\rho$. Thus, the minimal rate is maximized if and only if $\bar{R}_{\mathrm{S2V}}=\bar{R}_{\mathrm{V2D}}$. This is an equation for variable $\rho$. When $\frac{N_{\mathrm{S}}^{\mathrm{tot}} N_{\mathrm{r}}^{\mathrm{tot}} P_{\mathrm{S}}^{\mathrm{tot}}\sigma_{2}^{2}}{N_{\mathrm{t}}^{\mathrm{tot}} N_{\mathrm{D}}^{\mathrm{tot}} P_{\mathrm{V}}^{\mathrm{tot}}\sigma_{1}^{2}}=1$, we obtain a linear equation with solution $\rho^{\star}=\frac{1}{2}$. For the other cases, we have a quadratic equation with solution $\rho^{\star}=\frac{-b'-\sqrt{b'^2-4a'c'}}{2a'}$ as shown in \eqref{OptPosition}, which is the unique solution located in the interval $[0,1]$. This completes the proof.

\section{Proof of Theorem \ref{Theo_modulus}}
For notational simplicity, we employ the definitions $\mathbf{h}_{\mathrm{S2V}}=\mathbf{H}_{\mathrm{S2V}}\mathbf{w}_{\mathrm{S}}^{(k-1)}$ and $\mathbf{h}_{\mathrm{SI}}=\mathbf{H}_{\mathrm{SI}}\mathbf{w}_{\mathrm{t}}^{(k-1)}$ in \eqref{eq_problem_sub1}.
Note that Problems \eqref{eq_problem_sub1}, \eqref{eq_problem_sub2}, \eqref{eq_problem_sub3}, and \eqref{eq_problem_sub4} have a similar form, Theorem \ref{Theo_modulus} holds for all four problems. We only present the proof for Problem \eqref{eq_problem_sub1}. A similar proof can be provided for the other problems.

Let $\mathbf{w}_{\mathrm{r}}^{\circ}$ denote the optimal solution of Problem \eqref{eq_problem_sub1}, which satisfies
\begin{equation}
\left\{
\begin{aligned}
&\mathbf{w}_{\mathrm{r}}^{\circ\mathrm{H}} \mathbf{h}_{\mathrm{S2V}} = l_{1}e^{j\omega_{1}} \\
&\mathbf{w}_{\mathrm{r}}^{\circ\mathrm{H}}\mathbf{h}_{\mathrm{SI}} = l_{2}e^{j\omega_{2}},
\end{aligned}
\right.
\end{equation}
where $l_{1}$ and $\omega_{1}$ denote the modulus and phase of $\mathbf{w}_{\mathrm{r}}^{\circ\mathrm{H}} \mathbf{h}_{\mathrm{S2V}}$, respectively. $l_{2}$ and $\omega_{2}$ denote the modulus and phase of $\mathbf{w}_{\mathrm{r}}^{\circ\mathrm{H}}\mathbf{h}_{\mathrm{SI}}$, respectively. According to the formulation of Problem \eqref{eq_problem_sub1}, we know that $l_{2} \leq \eta_{1}^{(k)}$ and $l_{1}$ is the maximum of the objective function.

Note that $N_{\mathrm{r}}^{\mathrm{tot}}\geq 2$ is an implicit precondition for beamforming at the mmWave FD-UAV relay. Assume that $\mathbf{w}_{\mathrm{r}}^{\circ}$ has two elements which do not satisfy the CM constraint, i.e., $\left |[\mathbf{w}_{\mathrm{r}}^{\circ}]_{\pi_1}\right |<\frac{1}{\sqrt{N_{\mathrm{r}}^{\mathrm{tot}}}}$ and $\left |[\mathbf{w}_{\mathrm{r}}^{\circ}]_{\pi_2}\right |<\frac{1}{\sqrt{N_{\mathrm{r}}^{\mathrm{tot}}}}$, where $\{\pi_{n}\}\subseteqq\{1,2,\cdots,N_{\mathrm{r}}^{\mathrm{tot}}\}$ is the sequence of the BFV's indices. Furthermore, we keep $[\mathbf{w}_{\mathrm{r}}]_{\pi_n}=[\mathbf{w}_{\mathrm{r}}^{\circ}]_{\pi_n}$ fixed for $n=3,4,\cdots,N_{\mathrm{r}}^{\mathrm{tot}}$, and construct a new solution by adjusting $[\mathbf{w}_{\mathrm{r}}]_{\pi_1}$ and $[\mathbf{w}_{\mathrm{r}}]_{\pi_2}$, which can be obtained by solving the following problem:
\begin{equation}\label{eq_problem_sub1_pro}
\begin{aligned}
\mathop{\mathrm{Maximize}}\limits_{[\mathbf{w}_{\mathrm{r}}]_{\pi_1},[\mathbf{w}_{\mathrm{r}}]_{\pi_2}}~~~~~
&\left|\mathbf{w}_{\mathrm{r}}^{\mathrm{H}} \mathbf{h}_{\mathrm{S2V}}\right|\\
\mathrm{Subject~ to}~~~~~ &\mathbf{w}_{\mathrm{r}}^{\mathrm{H}}\mathbf{h}_{\mathrm{SI}} = l_{2}e^{j\omega_{2}}, \\
&\left|\left[ \mathbf{w}_{\mathrm{r}}\right]_{\pi_1}\right| \leq \frac{1}{\sqrt{N_{\mathrm{r}}^{\mathrm{tot}}}},\\
&\left|\left[ \mathbf{w}_{\mathrm{r}}\right]_{\pi_2}\right| \leq \frac{1}{\sqrt{N_{\mathrm{r}}^{\mathrm{tot}}}}.
\end{aligned}
\end{equation}
Based on the assumption that $\mathbf{w}_{\mathrm{r}}^{\circ}$ is the optimal solution of Problem \eqref{eq_problem_sub1}, we know that $\mathbf{w}_{\mathrm{r}}^{\circ}$ is also the optimal solution of Problem \eqref{eq_problem_sub1_pro}, because the feasible region of Problem \eqref{eq_problem_sub1_pro} is a subset of that of Problem \eqref{eq_problem_sub1}.

Next, we provide the following two lemmas to illustrate a key property of the solution, for $\frac{\left[\mathbf{h}_{\mathrm{S2V}}\right]_{\pi_1}}{\left[\mathbf{h}_{\mathrm{S2V}}\right]_{\pi_2}} \neq \frac{\left[\mathbf{h}_{\mathrm{SI}}\right]_{\pi_1}}{\left[\mathbf{h}_{\mathrm{SI}}\right]_{\pi_2}}$ and $\frac{\left[\mathbf{h}_{\mathrm{S2V}}\right]_{\pi_1}}{\left[\mathbf{h}_{\mathrm{S2V}}\right]_{\pi_2}} = \frac{\left[\mathbf{h}_{\mathrm{SI}}\right]_{\pi_1}}{\left[\mathbf{h}_{\mathrm{SI}}\right]_{\pi_2}}$, respectively.

\begin{lemma}
If $\frac{\left[\mathbf{h}_{\mathrm{S2V}}\right]_{\pi_1}}{\left[\mathbf{h}_{\mathrm{S2V}}\right]_{\pi_2}} \neq \frac{\left[\mathbf{h}_{\mathrm{SI}}\right]_{\pi_1}}{\left[\mathbf{h}_{\mathrm{SI}}\right]_{\pi_2}}$ holds, the assumption $\left |[\mathbf{w}_{\mathrm{r}}^{\circ}]_{\pi_1}\right |<\frac{1}{\sqrt{N_{\mathrm{r}}^{\mathrm{tot}}}}$ and $\left |[\mathbf{w}_{\mathrm{r}}^{\circ}]_{\pi_2}\right |<\frac{1}{\sqrt{N_{\mathrm{r}}^{\mathrm{tot}}}}$ cannot hold.
\end{lemma}

\begin{proof}
If $\frac{\left[\mathbf{h}_{\mathrm{S2V}}\right]_{\pi_1}}{\left[\mathbf{h}_{\mathrm{S2V}}\right]_{\pi_2}} \neq \frac{\left[\mathbf{h}_{\mathrm{SI}}\right]_{\pi_1}}{\left[\mathbf{h}_{\mathrm{SI}}\right]_{\pi_2}}$ holds, according to the first constraint in Problem \eqref{eq_problem_sub1_pro}, we can express $[\mathbf{w}_{\mathrm{r}}]_{\pi_2}$ as a function of $[\mathbf{w}_{\mathrm{r}}]_{\pi_1}$, i.e.,
\begin{equation}\label{eq_w2}
\begin{aligned}
\left[\mathbf{w}_{\mathrm{r}}\right]_{\pi_2}^{*}
&=\frac{l_{2}e^{j\omega_{2}}-\sum \limits_{n=3}^{N_{\mathrm{r}}^{\mathrm{tot}}} [\mathbf{w}_{\mathrm{r}}^{\circ}]_{\pi_n}^{*}\left[\mathbf{h}_{\mathrm{SI}}\right]_{\pi_n}}
{\left[\mathbf{h}_{\mathrm{SI}}\right]_{\pi_2}}
-[\mathbf{w}_{\mathrm{r}}]_{\pi_1}^{*}
\frac{\left[\mathbf{h}_{\mathrm{SI}}\right]_{\pi_1}}{\left[\mathbf{h}_{\mathrm{SI}}\right]_{\pi_2}}\\
&\triangleq f_{1}\left([\mathbf{w}_{\mathrm{r}}]_{\pi_1}\right).
\end{aligned}
\end{equation}
Substituting \eqref{eq_w2} into the objective function of Problem \eqref{eq_problem_sub1_pro}, we obtain
\begin{equation}\label{eq_obj_w1}\small
\begin{aligned}
&\mathbf{w}_{\mathrm{r}}^{\mathrm{H}} \mathbf{h}_{\mathrm{S2V}}\\
=&[\mathbf{w}_{\mathrm{r}}]_{\pi_1}^{*}\left[\mathbf{h}_{\mathrm{S2V}}\right]_{\pi_1}+[\mathbf{w}_{\mathrm{r}}]_{\pi_2}^{*}\left[\mathbf{h}_{\mathrm{S2V}}\right]_{\pi_2}
+\sum \limits_{n=3}^{N_{\mathrm{r}}^{\mathrm{tot}}} [\mathbf{w}_{\mathrm{r}}^{\circ}]_{\pi_n}^{*}\left[\mathbf{h}_{\mathrm{S2V}}\right]_{\pi_n}\\
=&[\mathbf{w}_{\mathrm{r}}]_{\pi_1}^{*}\underbrace{\left(\left[\mathbf{h}_{\mathrm{S2V}}\right]_{\pi_1}-\left[\mathbf{h}_{\mathrm{S2V}}\right]_{\pi_2}
\frac{\left[\mathbf{h}_{\mathrm{SI}}\right]_{\pi_1}}{\left[\mathbf{h}_{\mathrm{SI}}\right]_{\pi_2}}\right)} \limits_{=\hat{k}}\\
&~+\underbrace{\left[\mathbf{h}_{\mathrm{S2V}}\right]_{\pi_2}\frac{l_{2}e^{j\omega_{2}}-\sum \limits_{n=3}^{N_{\mathrm{r}}^{\mathrm{tot}}} [\mathbf{w}_{\mathrm{r}}^{\circ}]_{\pi_n}^{*}\left[\mathbf{h}_{\mathrm{SI}}\right]_{\pi_n}}
{\left[\mathbf{h}_{\mathrm{SI}}\right]_{\pi_2}}
\sum \limits_{n=3}^{N_{\mathrm{r}}^{\mathrm{tot}}} [\mathbf{w}_{\mathrm{r}}^{\circ}]_{\pi_n}^{*}\left[\mathbf{h}_{\mathrm{S2V}}\right]_{\pi_n}}
\limits_{=\hat{b}}\\
\triangleq & \hat{k}[\mathbf{w}_{\mathrm{r}}]_{\pi_1}^{*}+\hat{b} \triangleq f_{2}\left([\mathbf{w}_{\mathrm{r}}]_{\pi_1}\right).
\end{aligned}
\end{equation}

Note that $\frac{\left[\mathbf{h}_{\mathrm{S2V}}\right]_{\pi_1}}{\left[\mathbf{h}_{\mathrm{S2V}}\right]_{\pi_2}} \neq \frac{\left[\mathbf{h}_{\mathrm{SI}}\right]_{\pi_1}}{\left[\mathbf{h}_{\mathrm{SI}}\right]_{\pi_2}}$ holds in Lemma 1. Thus, we have $\hat{k}\neq 0$ in \eqref{eq_obj_w1}. Because of the assumption $\left |[\mathbf{w}_{\mathrm{r}}^{\circ}]_{\pi_1}\right |<\frac{1}{\sqrt{N_{\mathrm{r}}^{\mathrm{tot}}}}$ and $\left |[\mathbf{w}_{\mathrm{r}}^{\circ}]_{\pi_2}\right |<\frac{1}{\sqrt{N_{\mathrm{r}}^{\mathrm{tot}}}}$, we can always find a real number $\delta$, which is positive and small enough to satisfy
\begin{equation}
\left\{
\begin{aligned}
&\left |[\mathbf{w}_{\mathrm{r}}^{\circ}]_{\pi_1} \pm \delta\right |<\frac{1}{\sqrt{N_{\mathrm{r}}^{\mathrm{tot}}}},\\
&\left |f_{1}\left([\mathbf{w}_{\mathrm{r}}^{\circ}]_{\pi_1} \pm \delta\right)\right |<\frac{1}{\sqrt{N_{\mathrm{r}}^{\mathrm{tot}}}}.
\end{aligned}
\right.
\end{equation}
This means that $\left([\mathbf{w}_{\mathrm{r}}^{\circ}]_{\pi_1}+\delta\right)$ and $\left([\mathbf{w}_{\mathrm{r}}^{\circ}]_{\pi_1}-\delta\right)$ are both located in the feasible region of Problem \eqref{eq_problem_sub1_pro}. Since $[\mathbf{w}_{\mathrm{r}}^{\circ}]_{\pi_1}$ is the optimal solution of Problem \eqref{eq_problem_sub1_pro}, the objective function at $[\mathbf{w}_{\mathrm{r}}^{\circ}]_{\pi_1}+\delta$ and $[\mathbf{w}_{\mathrm{r}}^{\circ}]_{\pi_1}-\delta$ is no larger than at $[\mathbf{w}_{\mathrm{r}}^{\circ}]_{\pi_1}$, i.e.,
\begin{equation}
\left\{
\begin{aligned}
&\left |f_{2}\left([\mathbf{w}_{\mathrm{r}}^{\circ}]_{\pi_1}+\delta \right)\right |^{2} \leq \left |f_{2}\left([\mathbf{w}_{\mathrm{r}}^{\circ}]_{\pi_1}\right)\right |^{2},\\
&\left |f_{2}\left([\mathbf{w}_{\mathrm{r}}^{\circ}]_{\pi_1}-\delta \right)\right |^{2} \leq \left |f_{2}\left([\mathbf{w}_{\mathrm{r}}^{\circ}]_{\pi_1}\right)\right |^{2},
\end{aligned}
\right.
\end{equation}
According to the definition in \eqref{eq_obj_w1}, we obtain
\begin{equation}\small
\begin{aligned}
&\left\{
\begin{aligned}
&\left |\hat{k}[\mathbf{w}_{\mathrm{r}}^{\circ}]_{\pi_1}^{*}+\hat{b}+\hat{k}\delta\right |^{2} \leq \left |\hat{k}[\mathbf{w}_{\mathrm{r}}^{\circ}]_{\pi_1}^{*}+\hat{b}\right |^{2}\\
&\left |\hat{k}[\mathbf{w}_{\mathrm{r}}^{\circ}]_{\pi_1}^{*}+\hat{b}-\hat{k}\delta\right |^{2} \leq \left |\hat{k}[\mathbf{w}_{\mathrm{r}}^{\circ}]_{\pi_1}^{*}+\hat{b}\right |^{2}\\
\end{aligned}
\right.
\Rightarrow \\& \left\{
\begin{aligned}
&\mathfrak{R}\left(\left(\hat{k}[\mathbf{w}_{\mathrm{r}}^{\circ}]_{\pi_1}^{*}+\hat{b}\right)^{*}\hat{k}\delta\right)+
\left|\hat{k}\delta\right |^{2} \leq 0\\
&-\mathfrak{R}\left(\left(\hat{k}[\mathbf{w}_{\mathrm{r}}^{\circ}]_{\pi_1}^{*}+\hat{b}\right)^{*}\hat{k}\delta\right)+
\left|\hat{k}\delta\right |^{2} \leq 0\\
\end{aligned}
\right.
\Rightarrow 2\left|\hat{k}\delta\right |^{2}\leq 0,
\end{aligned}
\end{equation}
which contradicts the fact that $\hat{k}\neq 0$ and $\delta>0$. Thus, we can conclude that the assumption that $\mathbf{w}_{\mathrm{r}}^{\circ}$ has two elements that do not satisfy the CM constraint cannot hold when $\frac{\left[\mathbf{h}_{\mathrm{S2V}}\right]_{\pi_1}}{\left[\mathbf{h}_{\mathrm{S2V}}\right]_{\pi_2}} \neq \frac{\left[\mathbf{h}_{\mathrm{SI}}\right]_{\pi_1}}{\left[\mathbf{h}_{\mathrm{SI}}\right]_{\pi_2}}$. In other words, if there are any two elements that do not satisfy the CM constraint, they always have $\frac{\left[\mathbf{h}_{\mathrm{S2V}}\right]_{\pi_1}}{\left[\mathbf{h}_{\mathrm{S2V}}\right]_{\pi_2}} = \frac{\left[\mathbf{h}_{\mathrm{SI}}\right]_{\pi_1}}{\left[\mathbf{h}_{\mathrm{SI}}\right]_{\pi_2}}$.
\end{proof}

\begin{lemma}
If $\frac{\left[\mathbf{h}_{\mathrm{S2V}}\right]_{\pi_1}}{\left[\mathbf{h}_{\mathrm{S2V}}\right]_{\pi_2}} = \frac{\left[\mathbf{h}_{\mathrm{SI}}\right]_{\pi_1}}{\left[\mathbf{h}_{\mathrm{SI}}\right]_{\pi_2}}$ holds, there always exists another optimal solution of Problem \eqref{eq_problem_sub1_pro}, where at least one of $[\mathbf{w}_{\mathrm{r}}]_{\pi_1}$ and $[\mathbf{w}_{\mathrm{r}}]_{\pi_2}$  satisfies the CM constraint.
\end{lemma}

\begin{proof}
Based on $\frac{\left[\mathbf{h}_{\mathrm{S2V}}\right]_{\pi_1}}{\left[\mathbf{h}_{\mathrm{S2V}}\right]_{\pi_2}} = \frac{\left[\mathbf{h}_{\mathrm{SI}}\right]_{\pi_1}}{\left[\mathbf{h}_{\mathrm{SI}}\right]_{\pi_2}}$, we obtain
\begin{equation}
\begin{aligned}
&\frac{\left[\mathbf{h}_{\mathrm{S2V}}\right]_{\pi_1}}{\left[\mathbf{h}_{\mathrm{SI}}\right]_{\pi_1}} =\frac{\left[\mathbf{h}_{\mathrm{S2V}}\right]_{\pi_2}}{\left[\mathbf{h}_{\mathrm{SI}}\right]_{\pi_2}}\\
=&\frac{[\mathbf{w}_{\mathrm{r}}]_{\pi_1}^{*}\left[\mathbf{h}_{\mathrm{S2V}}\right]_{\pi_1}+
[\mathbf{w}_{\mathrm{r}}]_{\pi_2}^{*}\left[\mathbf{h}_{\mathrm{S2V}}\right]_{\pi_2}}{[\mathbf{w}_{\mathrm{r}}]_{\pi_1}^{*}
\left[\mathbf{h}_{\mathrm{SI}}\right]_{\pi_1}+
[\mathbf{w}_{\mathrm{r}}]_{\pi_2}^{*}\left[\mathbf{h}_{\mathrm{SI}}\right]_{\pi_2}}
\triangleq \chi.
\end{aligned}
\end{equation}
This indicates that $[\mathbf{w}_{\mathrm{r}}]_{\pi_1}^{*}\left[\mathbf{h}_{\mathrm{S2V}}\right]_{\pi_1}+
[\mathbf{w}_{\mathrm{r}}]_{\pi_2}^{*}\left[\mathbf{h}_{\mathrm{S2V}}\right]_{\pi_2}$ and $[\mathbf{w}_{\mathrm{r}}]_{\pi_1}^{*}\left[\mathbf{h}_{\mathrm{SI}}\right]_{\pi_1}+
[\mathbf{w}_{\mathrm{r}}]_{\pi_2}^{*}\left[\mathbf{h}_{\mathrm{SI}}\right]_{\pi_2}$ always have the same ratio regardless of the values of $[\mathbf{w}_{\mathrm{r}}]_{\pi_1}$ and $[\mathbf{w}_{\mathrm{r}}]_{\pi_2}$. We call this property the \emph{constant-ratio property}.

Since $\left |[\mathbf{w}_{\mathrm{r}}^{\circ}]_{\pi_1}\right |<\frac{1}{\sqrt{N_{\mathrm{r}}^{\mathrm{tot}}}}$ and $\left |[\mathbf{w}_{\mathrm{r}}^{\circ}]_{\pi_2}\right |<\frac{1}{\sqrt{N_{\mathrm{r}}^{\mathrm{tot}}}}$, it is easy to see that
\begin{equation}
\begin{aligned}
0 &\leq \left |[\mathbf{w}_{\mathrm{r}}^{\circ}]_{\pi_1}^{*}\left[\mathbf{h}_{\mathrm{S2V}}\right]_{\pi_1}+
[\mathbf{w}_{\mathrm{r}}^{\circ}]_{\pi_2}^{*}\left[\mathbf{h}_{\mathrm{S2V}}\right]_{\pi_2} \right |\\
& < \frac{1}{\sqrt{N_{\mathrm{r}}^{\mathrm{tot}}}}\left(|\left[\mathbf{h}_{\mathrm{S2V}}\right]_{\pi_1}|+|\left[\mathbf{h}_{\mathrm{S2V}}\right]_{\pi_2}|\right),
\end{aligned}
\end{equation}
and
\begin{equation}
\begin{aligned}
0 &\leq \left |[\mathbf{w}_{\mathrm{r}}^{\circ}]_{\pi_1}^{*}\left[\mathbf{h}_{\mathrm{SI}}\right]_{\pi_1}+
[\mathbf{w}_{\mathrm{r}}^{\circ}]_{\pi_2}^{*}\left[\mathbf{h}_{\mathrm{SI}}\right]_{\pi_2} \right |\\
& < \frac{1}{\sqrt{N_{\mathrm{r}}^{\mathrm{tot}}}}\left(|\left[\mathbf{h}_{\mathrm{SI}}\right]_{\pi_1}|
+|\left[\mathbf{h}_{\mathrm{SI}}\right]_{\pi_2}|\right).
\end{aligned}
\end{equation}

\begin{figure}[t]
\begin{center}
  \includegraphics[width=\figwidth cm]{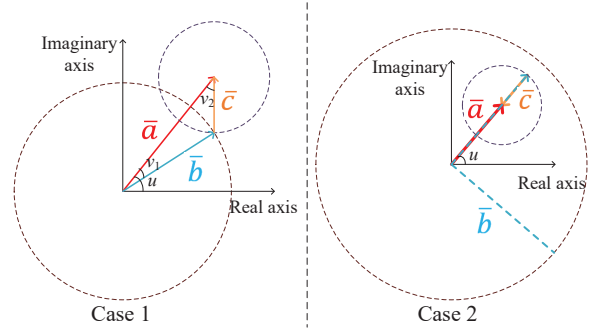}
  \caption{Illustration of the adjustment for the BFV's elements.}
  \label{Fig:w_adj}
\end{center}
\end{figure}
Next, we will consider two cases shown in Fig. \ref{Fig:w_adj}. We define $\bar{a}=\left|[\mathbf{w}_{\mathrm{r}}^{\circ}]_{\pi_1}^{*}\left[\mathbf{h}_{\mathrm{S2V}}\right]_{\pi_1}+
[\mathbf{w}_{\mathrm{r}}^{\circ}]_{\pi_2}^{*}\left[\mathbf{h}_{\mathrm{S2V}}\right]_{\pi_2}\right|$, $\bar{b}=\frac{1}{\sqrt{N_{\mathrm{r}}^{\mathrm{tot}}}}|\left[\mathbf{h}_{\mathrm{S2V}}\right]_{\pi_1}|$, and $\bar{c}=\frac{1}{\sqrt{N_{\mathrm{r}}^{\mathrm{tot}}}}|\left[\mathbf{h}_{\mathrm{S2V}}\right]_{\pi_2}|$. The corresponding angles in Fig. \ref{Fig:w_adj} are defined as follows
\begin{equation}
\left\{
\begin{aligned}
&u=\angle \left([\mathbf{w}_{\mathrm{r}}^{\circ}]_{\pi_1}^{*}\left[\mathbf{h}_{\mathrm{S2V}}\right]_{\pi_1}+
[\mathbf{w}_{\mathrm{r}}^{\circ}]_{\pi_2}^{*}\left[\mathbf{h}_{\mathrm{S2V}}\right]_{\pi_2}\right),\\
&v_{1}=\arccos\frac{\bar{a}^2+\bar{b}^2-\bar{c}^2}{2\bar{a}\bar{b}},\\
&v_{2}=\arccos\frac{\bar{a}^2+\bar{c}^2-\bar{b}^2}{2\bar{a}\bar{c}}.
\end{aligned}
\right.
\end{equation}

\emph{Case 1}: $\bar{a}\geq \left|\bar{b}-\bar{c}\right|$.

In this case, according to the constant-ratio property, it is easy to verify that $\left |[\mathbf{w}_{\mathrm{r}}^{\circ}]_{\pi_1}^{*}\left[\mathbf{h}_{\mathrm{SI}}\right]_{\pi_1}+
[\mathbf{w}_{\mathrm{r}}^{\circ}]_{\pi_2}^{*}\left[\mathbf{h}_{\mathrm{SI}}\right]_{\pi_2} \right |
\geq \frac{1}{\sqrt{N_{\mathrm{r}}^{\mathrm{tot}}}}\left(\left|\left[\mathbf{h}_{\mathrm{SI}}\right]_{\pi_1}|
+|\left[\mathbf{h}_{\mathrm{SI}}\right]_{\pi_2}\right|\right)$ holds. According to the triangle inequality, we can always find other $[\mathbf{w}_{\mathrm{r}}]_{\pi_1}$ and $[\mathbf{w}_{\mathrm{r}}]_{\pi_2}$ which satisfy the CM constraint. The basic idea is to adjust the phases of the two complex elements, and keep $[\mathbf{w}_{\mathrm{r}}]_{\pi_1}^{*}\left[\mathbf{h}_{\mathrm{S2V}}\right]_{\pi_1}+
[\mathbf{w}_{\mathrm{r}}]_{\pi_2}^{*}\left[\mathbf{h}_{\mathrm{S2V}}\right]_{\pi_2}=\bar{a}e^{ju}$ unchanged in Fig. \ref{Fig:w_adj}. The new solutions are generated as follows
\begin{equation}\label{eq_w_adj}
\left\{
\begin{aligned}
&[\mathbf{w}_{\mathrm{r}}^{\diamond}]_{\pi_1}=\frac{1}{\sqrt{N_{\mathrm{r}}^{\mathrm{tot}}}}e^{-j(u-v_{1}-\vartheta_{1})},\\
&[\mathbf{w}_{\mathrm{r}}^{\diamond}]_{\pi_2}=\frac{1}{\sqrt{N_{\mathrm{r}}^{\mathrm{tot}}}}e^{-j(u+v_{2}-\vartheta_{2})},
\end{aligned}
\right.
\end{equation}
where $\vartheta_{1}=\angle (\left[\mathbf{h}_{\mathrm{S2V}}\right]_{\pi_1})$ and $\vartheta_{2}=\angle (\left[\mathbf{h}_{\mathrm{S2V}}\right]_{\pi_2})$.
Then, it is easy to verify that $[\mathbf{w}_{\mathrm{r}}^{\diamond}]_{\pi_1}$ and $[\mathbf{w}_{\mathrm{r}}^{\diamond}]_{\pi_2}$ in \eqref{eq_w_adj} satisfy

\begin{equation}\label{eq_invir}
\left\{
\begin{aligned}
&[\mathbf{w}_{\mathrm{r}}^{\diamond}]_{\pi_1}^{*}\left[\mathbf{h}_{\mathrm{S2V}}\right]_{\pi_1}+
[\mathbf{w}_{\mathrm{r}}^{\diamond}]_{\pi_2}^{*}\left[\mathbf{h}_{\mathrm{S2V}}\right]_{\pi_2} \\ &~~~~~~=[\mathbf{w}_{\mathrm{r}}^{\circ}]_{\pi_1}^{*}\left[\mathbf{h}_{\mathrm{S2V}}\right]_{\pi_1}
+[\mathbf{w}_{\mathrm{r}}^{\circ}]_{\pi_2}^{*}\left[\mathbf{h}_{\mathrm{S2V}}\right]_{\pi_2},\\
&[\mathbf{w}_{\mathrm{r}}^{\diamond}]_{\pi_1}^{*}\left[\mathbf{h}_{\mathrm{SI}}\right]_{\pi_1}+
[\mathbf{w}_{\mathrm{r}}^{\diamond}]_{\pi_2}^{*}\left[\mathbf{h}_{\mathrm{SI}}\right]_{\pi_2}\\
&~~~~~~=[\mathbf{w}_{\mathrm{r}}^{\circ}]_{\pi_1}^{*}\left[\mathbf{h}_{\mathrm{SI}}\right]_{\pi_1}+
[\mathbf{w}_{\mathrm{r}}^{\circ}]_{\pi_2}^{*}\left[\mathbf{h}_{\mathrm{SI}}\right]_{\pi_2},
\end{aligned}
\right.
\end{equation}
which means that the designed $[\mathbf{w}_{\mathrm{r}}^{\diamond}]_{\pi_1}$ and $[\mathbf{w}_{\mathrm{r}}^{\diamond}]_{\pi_2}$ in \eqref{eq_w_adj} are also optimal solutions of Problem \eqref{eq_problem_sub1_pro} for which all elements satisfy the CM constraint.

\emph{Case 2}: $\bar{a} > \left|\bar{b}-\bar{c}\right|$.

In this case, according to the constant-ratio property, it is easy to verify that $\left |[\mathbf{w}_{\mathrm{r}}^{\circ}]_{\pi_1}^{*}\left[\mathbf{h}_{\mathrm{SI}}\right]_{\pi_1}+
[\mathbf{w}_{\mathrm{r}}^{\circ}]_{\pi_2}^{*}\left[\mathbf{h}_{\mathrm{SI}}\right]_{\pi_2} \right |
< \frac{1}{\sqrt{N_{\mathrm{r}}^{\mathrm{tot}}}}\left(\left|\left[\mathbf{h}_{\mathrm{SI}}\right]_{\pi_1}|
+|\left[\mathbf{h}_{\mathrm{SI}}\right]_{\pi_2}\right|\right)$ holds. This indicates that $[\mathbf{w}_{\mathrm{r}}]_{\pi_1}$ and $[\mathbf{w}_{\mathrm{r}}]_{\pi_2}$ cannot be adjusted such that both satisfy the CM constraint because the triangle inequality is not satisfied, i.e., the difference between the lengths of two sides is less than the length of the third side. However, we can adjust them such that one element satisfies the CM constraint. The basic idea is to enlarge the shorter side to satisfy the CM constraint, and then adjust the longer side to keep $[\mathbf{w}_{\mathrm{r}}]_{\pi_1}^{*}+
[\mathbf{w}_{\mathrm{r}}]_{\pi_2}^{*}\left[\mathbf{h}_{\mathrm{S2V}}\right]_{\pi_2}=\bar{a}e^{ju}$ unchanged in Fig. \ref{Fig:w_adj}.

Without loss of generality, we assume $\bar{b}\geq \bar{c}$ as shown in Fig. \ref{Fig:w_adj}.\footnote{When $\bar{b}< \bar{c}$, we can construct new optimal solutions in a similar manner.} Then, we can generate a new solution as follows
\begin{equation}\label{eq_w_adj1}
\left\{
\begin{aligned}
&[\mathbf{w}_{\mathrm{r}}^{\diamond}]_{\pi_1}=\left(\frac{1}{\sqrt{N_{\mathrm{r}}^{\mathrm{tot}}}}+\frac{a}{\left|\left[\mathbf{h}_{\mathrm{S2V}}\right]_{\pi_1}\right|}\right)e^{-j(u-\vartheta_{1})},\\
&[\mathbf{w}_{\mathrm{r}}^{\diamond}]_{\pi_2}=\frac{1}{\sqrt{N_{\mathrm{r}}^{\mathrm{tot}}}}e^{-j(u-\vartheta_{2}+\pi)}.
\end{aligned}
\right.
\end{equation}

It is easy to verify that $[\mathbf{w}_{\mathrm{r}}^{\diamond}]_{\pi_1}$ and $[\mathbf{w}_{\mathrm{r}}^{\diamond}]_{\pi_2}$ in \eqref{eq_w_adj1} satisfy \eqref{eq_invir}, which means that they are also an optimal solution of Problem \eqref{eq_problem_sub1_pro} for which only one element does not satisfy the CM constraint. Thus, we can conclude that if $\frac{\left[\mathbf{h}_{\mathrm{S2V}}\right]_{\pi_1}}{\left[\mathbf{h}_{\mathrm{S2V}}\right]_{\pi_2}} = \frac{\left[\mathbf{h}_{\mathrm{SI}}\right]_{\pi_1}}{\left[\mathbf{h}_{\mathrm{SI}}\right]_{\pi_2}}$ holds, we can always construct an optimal solution of Problem \eqref{eq_problem_sub1_pro}, where at most one element does not satisfy the CM constraint.
\end{proof}

Based on Lemma 1, we know that for any two elements of the BFV which do not satisfy the CM constraint, $\frac{\left[\mathbf{h}_{\mathrm{S2V}}\right]_{\pi_1}}{\left[\mathbf{h}_{\mathrm{S2V}}\right]_{\pi_2}} \neq \frac{\left[\mathbf{h}_{\mathrm{SI}}\right]_{\pi_1}}{\left[\mathbf{h}_{\mathrm{SI}}\right]_{\pi_2}}$ cannot hold. In other words, these elements always satisfy the \emph{constant-ratio property} in Lemma 2. Then, for any two elements that do not satisfy the CM constraint, we can always construct a new solution based on Lemma 2, where at most one element does not satisfies the CM constraint. Note that if there are three or more elements that do not satisfy the CM constraint, this construction can be repeated until only one or zero elements do not satisfy the CM constraint. Thus, we can conclude that there always exists an optimal solution of Problem \eqref{eq_problem_sub1}, for which at most one element of the optimal BFV does not satisfy the CM constraint.

\bibliographystyle{IEEEtran} 
\bibliography{IEEEabrv,Xiao60GHz,Xiao5GnNOMA}

\end{document}